%% file: paper_acm.tex
\documentclass[10pt, conference, letterpaper]{IEEEtran}
\pdfoutput=1

\usepackage[pass]{geometry}

\usepackage{graphicx}
\usepackage[font=footnotesize]{subcaption}
\usepackage{epstopdf}
\graphicspath{{figs/}}
\usepackage{amsthm}
\usepackage{amssymb}
\usepackage{amsmath}
\usepackage{verbatim}
\usepackage{cite}
\usepackage{url}
\usepackage{setspace}

\allowdisplaybreaks

\newcommand{\beq}{\begin{equation}}
\newcommand{\eeq}{\end{equation}}
\newcommand{\eg}{\textit{e.g.},}
\newcommand{\ie}{\textit{i.e.},}
\newcommand{\etal}{\textit{et~al.}}

\newcommand{\N}{\mathbb{N}}
\renewcommand{\P}{\mathbb{P}}

\newcommand{\stra}[1]{{\bf S{#1}}}
\def\bes2{\beta^{(s2)}}
\def\bes3{\beta^{(s3)}}
\def\ms2{\mu^{(s2)}}
\def\ms3{\mu^{(s3)}}

\newtheorem{theorem}{Theorem}
\newtheorem{lemma}{Lemma}

\usepackage[noend]{algpseudocode}
\usepackage{algorithm}

\title{Sharing LRU Cache Resources among Content Providers: A Utility-Based Approach}

\author{Mostafa Dehghan$^1$, Weibo Chu$^2$, Philippe Nain$^3$, Don Towsley$^1$\\
$^1${University of Massachusetts, Amherst, USA}\\
$^2${Northwestern Polytechnical University, Xi'an, China}, 
$^3${Inria, France}\\
{\tt \{mdehghan, towsley\}@cs.umass.edu, wbchu@nwpu.edu.cn, philippe.nain@inria.fr}
}

\begin{document}

\maketitle

\input{abstract}
\input{introduction}

\input{model}
\input{problem}
\input{online}
\input{simulation}
\input{discussion}
\input{related_work}
\input{conclusion}

\input{appendices}

\bibliographystyle{IEEEtran}
\bibliography{references}

\end{document}

%% file: abstract.tex
\begin{abstract}
In this paper, we consider the problem of allocating cache resources among multiple content providers. The cache can be partitioned into slices and each partition can be dedicated to a particular content provider, or shared among a number of them. It is assumed that each partition employs the LRU policy for managing content. We propose utility-driven partitioning, where we associate with each content provider a utility that is a function of the hit rate observed by the content provider. We consider two scenarios: i)~content providers serve disjoint sets of files, ii)~there is some overlap in the content served by multiple content providers. In the first case, we prove that cache partitioning outperforms cache sharing as cache size and numbers of contents served by providers go to infinity. In the second case, It can be beneficial to have separate partitions for overlapped content.  In the case of two providers it is usually always benefical to allocate a cache partition to serve all overlapped content and separate partitions to serve the non-overlapped contents of both providers.  We establish conditions when this is true asymptotically but also present an example where it is not true asymptotically.  We develop online algorithms that dynamically adjust partition sizes in order to maximize the overall utility and prove that they converge to optimal solutions, and through numerical evaluations we show they are effective.
\end{abstract}

%% file: introduction.tex
\section{Introduction}
The Internet has become a global information depository and content distribution platform, where various types of information or content are stored in the ``cloud'',  hosted by a wide array of {\em content providers}, and delivered or ``streamed'' on demand. 
The (nearly) ``anytime, anywhere access'' of online information or content -- especially multimedia content -- has precipitated rapid growth in Internet data traffic in recent years, both in wired and wireless (cellular) networks. It is estimated~\cite{cisco14} that the global Internet traffic in 2019 will reach 64 times its entire volume in 2005. A primary contributor to this rapid growth in data traffic comes from online video streaming services such as Netflix, Hulu, YouTube and Amazon Video, just to name a few. 
It was reported~\cite{sandvine-netflix} that Netflix alone consumed nearly a third of the peak downstream traffic in North America in 2012, and it is predicted~\cite{cisco14} that  nearly 90\% of all data traffic will come from video content distributors in the near future. 

Massive data traffic generated by large-scale online information access -- especially, ``over-the-top'' video delivery -- imposes an enormous burden on the Internet and poses many challenging issues. Storing, serving, and delivering videos to a large number of geographically dispersed users in particular require a vast and sophisticated infrastructure with huge computing, storage and network capacities. The challenges in developing and operating large-scale video streaming services in today's Internet~\cite{Vijay:YouTube-INFOCOM12,Vijay:Netflix,Vijay:Hulu} 
to handle user demands and meet user desired {\em quality-of-experience} also highlight some of the key limitations of today's Internet architecture.
This has led to a call for alternate Internet architectures that connect people to content rather than servers
(see~\cite{ahlgren2012survey} for a survey of representative architecture proposals). 
The basic premise of these content-oriented architectures is that storage is an integral part of the network substrate where content can be cached {\em on-the-fly}, or prefetched or ``staged'' {\em a priori}.

While there has been a flurry of recent research studies in the design of caching mechanisms~\cite{borst10,Carofiglio11,Zhang13,Michelle14}, relatively little attention has been paid to the problem of storage or {\em cache resource allocation among multiple content providers}. In this paper, we address a fundamental research question that is pertinent to all architecture designs:  {\em how to share or allocate the  cache resource within a single network forwarding element and  across various network forwarding elements among multiple content providers so as to maximize the cache resource utilization}  or {\em provide best utilities to content providers?}  

This question was addressed in \cite{ICN16} in an informal and heuristic manner.  It proposed a utility maximization framework, which we adopt, to address the aforementioned fundamental problem.
We consider a scenario where there are multiple content providers offering the same type of content, \eg~videos; the content objects offered by the content providers can be all distinct or there may be common objects owned by different content providers. Due to disparate user bases, the access probabilities of these content objects may vary across the CPs.
Our analysis and results are predicated on the use of Least Recently Used (LRU) cache replacement; however, we believe that they apply to other policies as well. 
\cite{ICN16} argued that, if all CPs offer {\em distinct} content objects, {\em partitioning the cache into slices of appropriate sizes, one slice per CP}, yields the best cache allocation strategy as it maximizes the sum of CP utilities. \cite{ICN16} also considered the case where CPs serve common content and argued that placing common content into a single LRU cache and non-common content in separate LRU caches usually provides the best performance.  We make more precise statements to support these observations.  In the case that common content are requested according to the same popularity distributions, regardless of provider, in the limit aggregate hit rate is maximized when three LRU partitions are established, one for the overlapped content and the other two for the non-overlap content.  We also provide a counterexample that shows that such a strategy is not always optimal.  However, the conclusion is that partitioning is usually best.

The above results are based on work of Fagin \cite{Fagin1977}, who characterized the asymptotic behavior of LRU for a rich class of content popularity distributions that include the Zipf distribution. 

In the last part of the paper, we develop decentralized algorithms to implement utility-driven cache partitioning. These algorithms adapt to changes in system parameters by dynamically adjusting the partition sizes, and are theoretically proven to be stable and converge to the optimal solution.

Our results illustrate the importance of considering the cache allocation problem among multiple CPs and has implications in architectural designs: from the perspective of cache resource efficiency or utility maximization of CPs, cache partitioning (among CPs) should be a basic principle for cache resource allocation; it also suggests alternate content-oriented architectures which explicitly account for the role of CPs~\cite{CONIA}.
Cache partitioning also provides a natural means to effectively handle heterogeneous types of content with different traffic or access characteristics, and offer differentiated services for content delivery~\cite{kelly99,ko03,lu04,Zhang13}.  In the future Internet  where network cache elements will likely be provided by various entities~\cite{OpenCDN},  our framework also facilitates the design of  distributed pricing and control mechanisms, and allows for the establishment of a viable cache market economic model.

The main contributions of this paper can be summarized as follows:
\begin{itemize}
\item We establish the connection between Fagin's asymptotic results on the LRU cach and the characteristic time (CT) approximation introduced in \cite{Che01}, providing a stronger theoretical underpinning for the latter than previously known.  Moreover, we extend Fagin's results and therefore theoretical justification of the CT approximation to a larger class of workloads that include those coming from independent content providers.

\item Using Fagin's asymptotic framework we show that partitioning is the best strategy for sharing a cache  when content providers do not have any content in common.  On the other hand when contnet providers serve the same content, it can be beneficial for content providers to share a cache to serve their overlapped content. We establish this to be true for a class of popularity distributions.  We also present an example where placing common content in a shared cache is not optimal.

\item We develop online algorithms for managing cache partitions, and prove the convergence of these algorithms to the optimal solution using Lyapunov functions.

\item We show that our framework can be used in revenue based models where content providers react to prices set by (cache) service providers without revealing their utility functions.

\item We perform simulations to show the efficiency of our online algorithms using different utility functions with different fairness implications.
\end{itemize}

The remainder of this paper is organized as follows. We present the problem setting and basic model in  Section~\ref{sec:model} where we make the connection between Fagin's asymptotic results and the CT approximation. We describe the cache allocation problem via the utility maximization framework in  Section~\ref{sec:problem}. In Section~\ref{sec:online}, we develop online algorithms for implementing utility-maximizing cache partitioning. Simulation results are presented in Section~\ref{sec:simulation}. In Section~\ref{sec:discussion}, we explore the implications of our results, and discuss future research directions and related work. Section~\ref{sec:conclusion} concludes the paper.

%% file: model.tex

\section{Problem Setting \& Basic Model}
\label{sec:model}
\begin{figure}[t]
\centering
  	\includegraphics[scale=0.19]{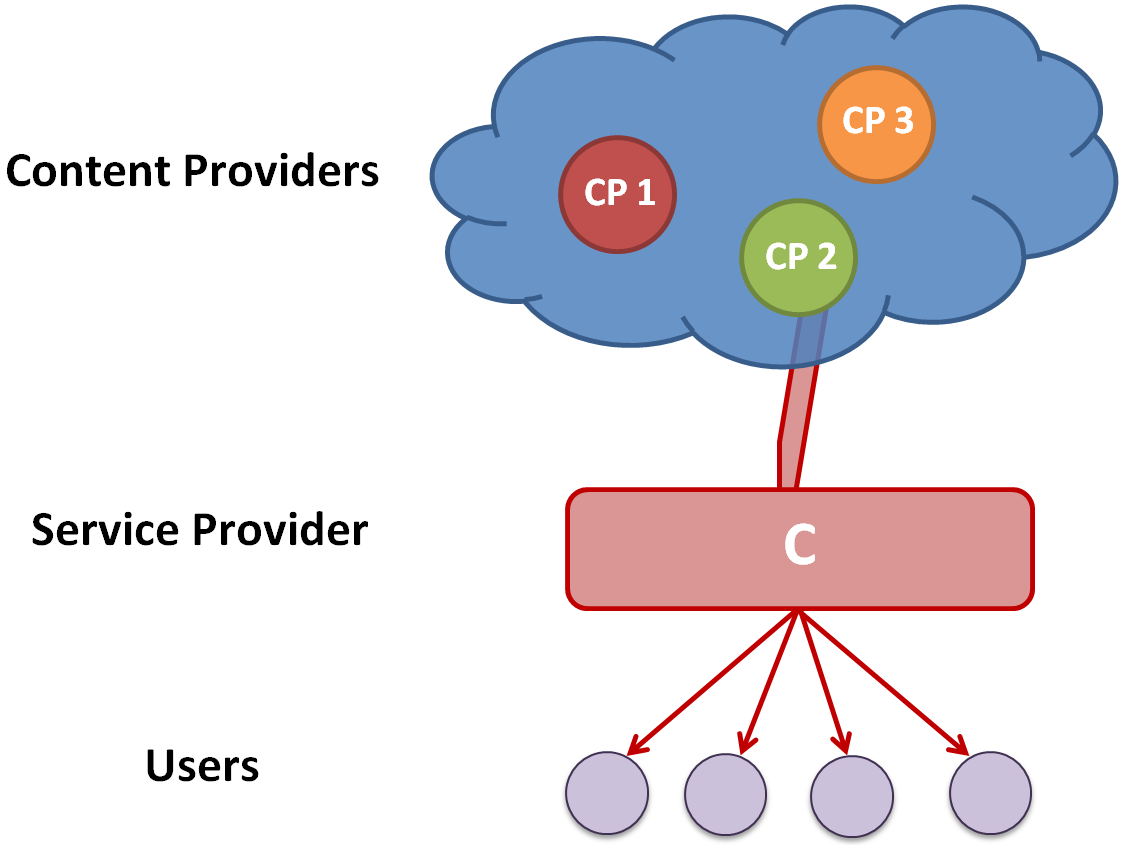} 
\caption{Network Model.}
    \label{fig:model}
\end{figure}

Consider a network as shown in Figure~\ref{fig:model}, where users access content, \eg~videos,  from $K$ content providers~(CPs).
CP~$k~(k=1,\ldots,K)$ serves a set $S_k$ of $n_k$ unit size files where $n_k = |S_k|$; we will usually label these files $i=1,\dotsc, n_k$. All CPs share a content cache, supplied by a third-party network provider, referred to as a service provider hereafter.
Content providers have business relations with the service provider and pay for cache resources. 
There are two possible scenarios: i)~the content objects offered by the CPs are all {\em distinct}; and ii)~some {\em common} objects are provided by different CPs. Due to disparate user bases, the access patterns of these content objects may vary across the CPs.

We assume that requests are described by a Poisson process with request rate for file $i$ of CP $k$ being ${\lambda_{k,i}=\lambda_k p_{k,i}, i\in S_k, k = 1,\ldots,K}$, where $\lambda_k$ denotes the aggregate request rate for contents from CP~$k$, and $p_{k,i}$ is the probability that a request to CP~$k$ is for content~$i$. Associated with each CP is a utility $U_k(h_k)$ that is an \emph{increasing} and \emph{concave} function of the hit rate $h_k$ over all its files. In its most general form, the service provider wishes to maximize the sum of the utilities over all content providers, $\sum_{k}{ U_k(h_k)}$, through a proper allocation of cache space to the CPs.  In the simple case where  $U_k(h_k)=h_k$, the objective becomes that of maximizing the overall cache hit rate, which provides a measure of  the overall cache utilization efficiency.

\vspace*{2pt}
\noindent
{\bf Cache Partitioning:}
When the cache is shared among the CPs, content objects offered by all CPs compete for the storage space based on their access patterns. To restrict cache contention to smaller sets of content, the service provider can form content groups from files served by a CP or multiple CPs, and partition the cache into slices and dedicate a partition to each content group. Let $P$ denote the number of content groups/cache partitions. Also, let $V_p$ and $C_p, p=1,\ldots,P$ denote the content groups and partition sizes, respectively. Note that $P=1$ implies that the cache is shared as a whole, while $P>1$ means it is partitioned.

The first question to ask is: \emph{what is the optimal number of partitions and how should files be grouped?}
To determine the number of slices and that what files should be requested from which partition, the service provider proceeds as follows. Files are first grouped into disjoint sets according to which content providers serve them. The service provider then decides how many partitions to create, and whether to dedicate a separate partition for each set of files, or have multiple sets of files share a partition. In the next section, we explain what could change if the cache manager made partitioning decisions on a per file basis rather than sets of files.

Assuming the answer to the first question, the second question is: \emph{how should the partitions be sized?}
Let ${\mathbf{C} = (C_1, C_2, \ldots, C_P)}$ denote the vector of partition sizes. For each content provider $k$, hit rate is a function of the partition sizes $h_k(\mathbf{C})$. For a cache of size $C$, we formulate this question as the following optimization problem:
\begin{align*}
\label{eq:main}
\text{maximize} \quad &\sum_{k=1}^{K}{U_k\Big(h_k(\mathbf{C})\Big)} \notag\\
\text{such that} \quad &\sum_{p=1}^{P}{C_p} \le C \\
& C_p = 0, 1, 2,\ldots; \quad p=1, 2, \ldots, P. \notag
\end{align*}

Note that the above formulation is an integer programming problem that is typically hard to solve. However, in practice caches are large and therefore we assume $C_p$ can take any real value, as the rounding error will be negligible.

\vspace*{2pt}
\noindent
{\bf Cache Characteristic Time:}
Assume a cache of size~$C$ serving~$n$ contents with popularity distribution~$p_i$, ${i=1,\ldots , n}$ .
Under the independent reference model (requests are i.i.d.), Fagin~\cite{Fagin1977} introduced the notion of a {\em window size} $T$ that satisfies
\[ C = \sum_{i=1}^n (1-(1-p_i)^T). \]
The miss probability associated with a window of size $T$ is defined as
\[ m(T) = \sum_{i=1}^n p_i(1-p_i)^T. \]
Fagin introduced a cumulative probability distribution, $F$, that is continuously differentiable in $(0,1)$ with $F(0) = 0$ and $F(1) = 1$. The right-derivative of $F$ at $0$, denoted by $F^\prime(0)$,
may be infinite. This will allow us to account for Zipf-like distributions.
Define 
\[ 
p_i^{(n)} = F(i/n) - F((i-1)/n), \quad i=1,\ldots , n
\]
the probability that page $i$ is requested. Hereafter, we will refer to $F$ as the popularity distribution.
If $C/n = \beta$ then $T/n \rightarrow \tau_0$ where (see~\cite{Fagin1977})
\beq \label{eq:beta}
 \beta = \int_0^1 (1-e^{-F'(x)\tau_0})dx 
\eeq
and $m(T) \rightarrow \mu$ where
\beq \label{eq:mu}
 \mu = \int_0^1 F'(x)e^{-F'(x)\tau_0}dx. 
\eeq
Moreover, $\mu$ is the limiting miss probability under LRU when~$n\rightarrow \infty$.

Suppose that requests arrive according to a Poisson process with rate $\lambda$. Express $\beta$ as
\[ \beta = \int_0^1 P(X(x) <\tau_0/\lambda) dx \]
where $X(x)$ is an exponential random variable with intensity~$\lambda F'(x)$.  $X(x)$ is the inter-arrival time of two requests for content of type $x$.  If this time is less than $\tau_0/\lambda$, then the request is served from the cache, otherwise it is not.   In practice, as $n$ is finite, this is approximated by 
\beq \label{eq:CT-beta}
 C = \beta n = \sum_{i=1}^n (1 - e^{-\lambda p^{(n)}_i T_c}) 
\eeq
where $T_c$ is the Characteristic Time (CT) for the finite content cache \cite{Che01}.  The aggregate miss probability is approximated by 
\beq \label{eq:CT-hit}
 m(T_c) =1- \sum_{i=1}^n p^{(n)}_i \bigl(1 - e^{-\lambda p_i^{(n)}T_c}\bigr).
 \eeq
Fagin's results suffice to show that as $n\rightarrow \infty$, the r.h.s. of~\eqref{eq:CT-hit} converges to the LRU miss probability.

In the context of $K$ providers, let $n_k = b_k n$, $n,b_k \in \N$, ${k=1,\ldots , K}$. Denote $B_k :=\sum_{j=1}^{k} b_j$ with $B_0 = 0$ by convention. It helps also to denote $B_K$ by $B$. Let $F_1,F_2, \ldots , F_K$ be continuous uniformly differentiable CDFs in $(0,1)$. $F_k^\prime(0)$ may be infinite for $k=1,\ldots K$. If each provider has a cache that can store a fraction $\beta_k$ of its contents, then the earlier described CT approximation, (\ref{eq:CT-beta}), (\ref{eq:CT-hit}), applies with 
\beq  \label{eq:popularity-k}
p_{k,i}^{(n)} = F_k\Bigl(\frac{i}{b_k n}\Bigr) 
  - F_k\Bigl(\frac{i-1}{b_kn}\Bigr), \quad i=1,\ldots,b_k n. 
\eeq
We denote the asymptotic miss probabilities for the $K$ caches, each using LRU, by 
\begin{equation}
\label{def-mup}
 \mu^{(p)}_k = \int_0^1 F'_k(x) e^{-F'_k(x) \tau_k}dx, \quad k=1,\dotsc , K
 \end{equation}
where $\tau_k$ is the solution of \eqref{eq:beta} with $\beta$ replaced by $\beta_k$.

Assume that the providers share a cache of size $C$.  Define $\beta^{(s)}=\sum_{k=1}^K \beta_k$. We introduce $\mu^{(s)}$ and $\tau_0$ through the following two equations,
\begin{align}
\mu^{(s)} & = \sum_{k=1}^K a_k \int_0^1 F'_k(x)e^{- a_k  F'_k(x)\tau_0 B/b_k}dx, 
      \label{eq:mu-shared} \\
\beta^{(s)} & = 1-\sum_{k=1}^K \frac{b_k}{B}\int_0^1 e^{-a_k  F'_k(x)\tau_0B/b_k}dx\label{eq:beta-shared}
\end{align}
where $a_k :=\lambda_k/\lambda$, $k=1,\ldots , K$. 

\begin{theorem}  
\label{th:shared_Fagin}
\label{THM:SHARED_FAGIN}
Assume that we have  $K$ providers with popularity distributions $F_1,\dotsc , F_K$ as defined above, with numbers of contents given by $b_k n$ and request rates $\lambda_k$, ${k=1,\dotsc , K}$. Construct the sequence of popularity probabilities $\{p_{k,i}^{(n)}\}$, $n=1, \dotsc$ defined in \eqref{eq:popularity-k} and cache sizes $C^{(n)}$ such that $C^{(n)}/n = \beta$.  Then, the aggregate miss probability under LRU  converges to $\mu^{(s)}$ given in \eqref{eq:mu-shared}, where $\tau_0$ is the unique solution of \eqref{eq:beta-shared}.
\end{theorem}
 \begin{proof}
See Appendix \ref{sec:shared_Fagin}.
\end{proof}
{\em Remark.} This extends Fagin's results to include any asymptotic popularity CDF $F$ that is continuously differentiable in~$(0,1)$ except at a countable number of points.

To help the reader with notation, a glossary of the main symbols used in this paper is given in Table~\ref{tbl:notation}.

\begin{table}[]
\centering
\caption{Glossary of notations.}
\hspace*{-0.5cm}
\begin{tabular}{ | c | l | }
\hline
$S_k$ & set of files served by content provider $k$ \\
$p_i$ & probability that file $i$ is requested \\
$\lambda_{k,i}$ & request rate for file $i$ of CP $k$ \\
$\lambda_k$ & total request rate for CP $k$ contents \\
$F(\cdot)$ & file popularity CDF \\
$h_k$ & hit rate of CP $k$ \\
$C_p$ & capacity of partition $p$ \\
$n$ & number of files \\
$\beta$ & normalized capacity \\
$\mu$ & limiting miss probability \\
\hline
\end{tabular}
\label{tbl:notation}
\end{table}

%% file: problem.tex

\section{Cache Resource Allocation among Content Providers}
\label{sec:problem}
In this section, we formulate cache management as a utility maximization problem. We introduce two formulations, one for the case where content providers serve distinct contents, and another one for the case where some contents are served by multiple providers.

\subsection{Content Providers with Distinct Objects}
Consider the case of $K$ providers with $n_k = b_kn$ contents each where $b_k, n \in \N$ and ${k=1,\ldots ,K}$. Also, let $B=\sum_{k=1}^K{b_k}$ Assume that requests to CP $k$ is characterized by a Poisson process with rate $\lambda_k$.

We ask the question whether the cache should be shared or partitioned between CPs under the LRU policy.  It is easy to construct cases where sharing the cache is beneficial.  However these arise when the cache size and the number of contents per CP are small.  Evidence suggests that partitioning provides a larger aggregate utility than sharing as cache size and number of contents grow. In fact, the following  theorem shows that asymptotically, under the assumptions of Theorem 1, in the limit as $n\rightarrow \infty$, the sum of utilities under LRU when the cache is partitioned, is at least as large as it is under LRU when the CPs share the cache. To do this, we formulate the following optimization problem: namely to partition the cache among the providers so as to maximize the sum of utilities:
\begin{align}
\label{eq:Up}
\max_{\beta_k} \,& U^{(p)}:= \sum_{k=1}^{K} U_k(\lambda_k (1-\mu_k(\beta_k))) \\
\text{s.t.} \hspace*{0.05in}  & \beta =  \sum_{k=1}^K \frac{b_k}{B}\beta_k, \notag\\
 & \beta_k  \ge 0, \quad k=1, 2, \ldots, K. \notag
\end{align}
Observe that  $\lambda_k (1-\mu_k(\beta_k))$ in (\ref{eq:Up}) is the hit rate of documents of CP $k$, where $\mu_k(\beta_k)$ is given by (\ref{def-mup}).

Here, $\mu_k$ is the asymptotic miss probability for content served by CP $k$, $\beta$ is the cache size constraint expressed in terms of the fraction of the aggregate content that can be cached, $b_k/B$ is the fraction of content belonging to provider~$k$, and $\beta_k$ is the fraction of CP $k$ content that is permitted in CP $k$'s partition. The direct dependence of $\beta_k$ on $\mu_k$ is difficult to capture.  Hence, we use \eqref{eq:beta} and \eqref{eq:mu}, to transform the above problem into:
\begin{align}
\max_{\tau_k} \, & U^{(p)}:= \sum_{k=1}^{K} U_k\Bigl(\lambda_k\bigl(1- \int_0^1 e^{-F'(x) \tau_k}dx \bigr)\Bigr) \\
\text{s.t.} \, & \beta =  1 - \sum_{k=1}^K \frac{b_k}{B}\int_0^1 e^{-F'_k(x)\tau_k}dx, \notag\\
 & \tau_k  \ge 0, \quad k=1, 2, \ldots, K. \notag
\end{align}

\begin{theorem}
\label{thm:distinct_partition}
Assume $K$ providers with popularity distributions constructed from 
distributions $F_1, \dotsc ,F_K$ using \eqref{eq:popularity-k} with number of contents $b_k n$ and request rates $\lambda_k$, ${k=1,\dotsc ,K}$ sharing a cache of size $C^{(n)}$ such that ${C^{(n)}/n=\beta}$. Then, as $n \rightarrow \infty$ the sum of utilities under partitioning is at least as large as that under sharing.
\end{theorem}
\begin{proof}
The sum of utilites for the shared case is
\[
U^{(s)} = \sum_{k=1}^{K} U_k\Bigl(\lambda_k \bigl(1-\int_0^1 e^{-F'(x) \tau_0}dx \bigr)\Bigr)
\]
where $\tau_0$ is the unique solution to 
\[
  \beta  = 1 - \sum_{k=1}^K \frac{b_k}{B}\int_0^1 e^{-F'_k(x)\tau_0}dx.
\]
When we set $\tau_k = a_k B \tau_0 / b_k$ in $U^{(p)}$ then $U^{(p)} = U^{(s)}$, proving the theorem.
\end{proof}

Based on the above theorem, we focus solely on partitioned caches and use the CT approximation to formulate  the utility--maximizing resource allocation problem for content providers with distinct files as follows:
\begin{align}
\label{eq:opt}
\text{maximize} \quad &\sum_{k=1}^{K}{U_k\Big(h_k(C_k)\Big)} \\
\text{such that} \quad &\sum_{k=1}^{K}{C_k} \le C, \notag\\
& C_k \ge 0, \quad k=1, 2, \ldots, K. \notag
\end{align}

In our formulation, we assume that each partition employs LRU for managing the cache content.  Therefore, we can compute the hit rate for content provider $k$ as
\begin{equation}
\label{hkCk}
h_k(C_k) = \lambda_k \sum_{i=1}^{n_k} p_{k,i}\bigl(1-e^{-\lambda_k p_{k,i} T_k(C_k)}\bigr),
\end{equation}
where $T_k(C_k)$ denotes the characteristic time of the partition with size $C_k$ dedicated to content provider $k$. $T_k(C_k)$ is the unique solution to the equation
\begin{equation}
\label{Ck}
C_k = \sum_{i=1}^{n_k} (1 - e^{-\lambda_k p_{k,i} T_k}).
\end{equation}
The following theorem establishes that resource allocation problem~\eqref{eq:opt} has a unique optimal solution:
\begin{theorem}
\label{thrm:opt_exists}
Given strictly concave utility functions, resource allocation problem~\eqref{eq:opt} has a unique optimal solution.
\end{theorem}

\begin{proof}
In Appendix~\ref{sec:hitrate_concavity}, we show that $h_k(C_k)$ is an increasing concave function of $C_k$. Since $U_k$ is assumed to be an increasing and strictly concave function of the cache hit rate $h_k$, it follows that $U_k$ is an increasing and strictly concave function of $C_k$. The objective function in~\eqref{eq:opt} is a linear combination of strictly concave functions, and hence is concave. Since the feasible solution set is convex, a unique maximizer called the optimal solution exists.
\end{proof}

Last, it is straightforward to show that partitioning is at least as good as sharing, for finite size systems using the CT approximation.
\subsection{Content Providers with Common Objects}
\begin{figure}[t]
\centering
  	\includegraphics[scale=0.24]{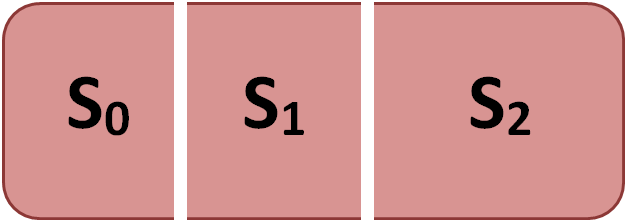}
\caption{Partitioning cache into three slices. One partition for the set of common files, $S_0$, and two other partitions, one for the remaining files from each content provider, $S_k$.}
    \label{fig:partition_shared}
\end{figure}
Here, we first assume there are only two content providers in the network and then consider the general case. There are three sets of content, $S_0$ of size $n_0$ served by both providers, and $S_1$ and $S_2$, sizes $n_1$ and $n_2$ served separately by each of the providers.  Requests are made to $S_0$ at rate $\lambda_{0,k}$ from provider $k$ and to $S_k$ at rate $\lambda_k$.  Given a request is made to $S_0$ from provider $k$, it is for content $i$ with probability $p_{0,k,i}$.  Similarly, if the request is for content in $S_k$, $k=1,2$, it is for content $i$ with probability $p_{k,i}$.

We have two conceivable cases for the files in $S_0$: 1)~each content provider needs to maintain its own copy of the content, \eg\ due to security reasons, or 2)~one copy can be kept in cache to serve requests to either of the content providers.
The first case can be treated as if there is no common content between the two content providers, and hence can be cast as problem~\eqref{eq:opt}.
For the second case, we consider three strategies for managing the cache:
\begin{itemize}
\item {\bf Strategy~1 (S1):} sharing the whole cache as one large partition,
\item {\bf Strategy~2 (S2):} partitioning into two dedicated slices, one for each CP,
\item {\bf Strategy~3 (S3):} partitioning into three slices, one shared partition for the set of common contents, and two other partitions for the remaining files of each CP, as shown in Figure~\ref{fig:partition_shared}.
\end{itemize}
The following theorem states that \stra{3} performs at least as well as \stra{1} in an asymptotic sense.
\begin{theorem}  \label{th:3over1} \label{THM:3OVER1}
Assume that we have two providers with a set of shared files $S_0$, and sets of non-shared files , $S_1,S_2$ with numbers of files  $n_k = b_k n$, $b_k,n \in \N$, and $k=0,1,2$. Assume that requests to these sets occur with rates $\lambda_{0,k}$ and $\lambda_k$ and content popularities are described by asymptotic popularity distributions $F_{0,k}$, and $F_k$ . Construct the sequence of popularity probabilities $\{p_{k,i}^{(n)}\}$, $n=1, \dotsc$ similar to \eqref{eq:popularity-k} and cache sizes $C^{(n)}$ such that $C^{(n)}/n = \beta$.  Then, the asymptotic aggregate LRU miss probability is at least as small under \stra{3} as under \stra{1}.
\end{theorem}
The proof is similar to the proof of Theorem~\ref{thm:distinct_partition}.

Neither \stra{2} nor \stra{3} outperforms the other for all problem instances, even asymptotically.  However, we present a class of workloads for which asymptotically \stra{3} outperforms \stra{2} and then follow it with an example where \stra{2} outperforms~\stra{3}.

Consider the following workload where the asymptotic popularity distributions of requests to the two providers for the shared content are identical, $F_{0,1} = F_{0,2}$.
\begin{theorem}\label{th:3over2} \label{THM:3OVER2}
Assume that we have two providers with a set of shared files $S_0$ and sets of non-shared files, $S_1,S_2$ with numbers of files $n_k = b_k n$, $b_k \in \N$, $n=1,\dotsc$, and $k=1,2,3$. Assume that threquests are described by Poisson processes with rates $\lambda_{0,k}$ and $\lambda_k$, $k=1,2$, and content popularities are described by asymptotic popularity distributions $F_{0,1} = F_{0,2}$, and $F_1, F_2$. Construct the sequence of popularity probabilities $\{p_{k,i}^{(n)}\}$, $n=1, \dotsc$ similar to \eqref{eq:popularity-k} and cache sizes $C^{(n)}$ such that $C^{(n)}/n = \beta$.  Then the asymptotic aggregate hit probability under LRU is at least as large under \stra{3} as under~\stra{2}.
\end{theorem}
 The proof is found in Appendix~\ref{sec:3over2}

Below is an example where \stra{2} outperforms \stra{3}.  The asymptotic popularity distributions for the shared content are given by
\[
F_{0,1}(x) = 
\begin{cases}
2x/11 & 0<x\le 1/2 \\
(20x -9)/11 & 1/2<x<1
\end{cases}
\]
and 
\[
F_{0,2}(x) = 
\begin{cases}
300x/151 & 0<x\le 1/2 \\
(2x +149)/151 & 1/2<x<1
\end{cases}
\]
with request rates $\lambda_{0,1} = 1.1$ and $\lambda_{0,2} = 15.1$. The asymptotic popularities of the non-shared contents are ${F_1(x) =F_2(x) = x}$ with request rates $\lambda_1 = 20$ and $\lambda_2 = 30$.  Last, there are equal numbers of content in each of these sets, $n_0=n_1=n_2$.  If we set $\beta = 2/3$, then the aggregate hit probability under \stra{3} with optimal partitioning is .804, which is slightly lower than the aggregate hit probability, .816, under \stra{2} with optimal partitioning. 

The above examples show that the workloads of the content providers can affect which strategy is optimal. However, we argue that partitioning into three slices should provide the best performance in most practical situations, where content providers have similar popularity patterns for the contents they commonly serve. This is unlike the second example where the two content providers have disparate rates for the common contents they serve. In Section~\ref{sec:simulation}, we will show that even if two content providers have dissimilar request rates for their common contents, partitioning into three slices does better.



Based on the above argument for the performance of partitioning into three slices in the case of two content providers, for $K$ content providers with common files, one should create a partition for each set of files that are served by a number of content providers. A procedure for creating the optimal set of partitions $\mathcal{P}$ with the files routed to each partition is given in Algorithm~\ref{alg:partition}. Algorithm~\ref{alg:partition} runs in $O(|S|^2)$ where $S$ denotes the set of all files served by all content providers. Note that the number of partitions can grow exponentially with the number of content providers.

\begin{algorithm}[t]
\caption{Partitioning a Cache serving $K$ content providers with possibility of common files among some content providers.}
\label{alg:partition}
\begin{algorithmic}[1]
	\Statex
		\State $S\gets S_1 \cup S_2 \cup \ldots \cup S_K$.
		\State $\mathcal{P}\gets \O$.
		\For{$f\in S$}
			\State $M_f\gets \{k:$ Content provider $k$ serves files $f\}$.
			\If{Exists $(V,M)\in P$ such that $M = M_f$}
				\State $V\gets V\cup \{f\}$.
			\Else
				\State $\mathcal{P}\gets \mathcal{P}\cup \{(\{f\}, M_f)\}$.
			\EndIf
		\EndFor
\end{algorithmic}
\end{algorithm}

Once the set of partitions $\mathcal{P}$ and the set of files corresponding to each partition is determined, the optimal partition sizes can be computed through the following optimization problem:
\begin{align}
\label{eq:opt3}
\text{maximize} \quad &\sum_k U_k(h_k) \\
\text{such that} \quad &\sum_{p=1}^{|\mathcal{P}|}{C_p} \le C \notag\\
& C_p \ge 0,  \quad p=1, 2, \ldots, |\mathcal{P}|, \notag
\end{align}
where hit rate for CP $k$ is computed as
\[h_k = \sum_{p=1}^{|\mathcal{P}|}\lambda_k \sum_{i\in V_p} p_{k,i}\bigl(1-e^{-\lambda_i p_{k,i} T_p(C_p)}\bigr)
\]
where  $V_p$ denotes the set of files requested from partition $p$, and $\lambda_i\triangleq \sum_k{\lambda_{k,i}}$ denotes the aggregate request rate for content $i$ through all content providers, and $T_p$ denotes the characteristic time of partition $p$.

\begin{theorem}
\label{THM:UNIQUE}
Given strictly concave utility functions, resource allocation problem~\eqref{eq:opt3} has a unique optimal solution.
\end{theorem}
\begin{proof}
In Appendix~\ref{sec:proof_obj_shared}, we show that the optimization problem~\eqref{eq:opt3} has a concave objective function. Since the feasible solution set is convex, a unique maximizer exists.
\end{proof}

\subsection{Implications}
\vspace*{1pt}
\noindent
\noindent
{\bf Cache Partitioning: Number of Slices, Management Complexity and
Static Caching.}
In our utility maximization formulations~\eqref{eq:opt} and~\eqref{eq:opt3} and their solution, the cache is only partitioned and allocated per CP for a set of {\em distinct} content objects owned by the CP; a cache slice is allocated and shared among several CPs  only for a set of {\em common} content objects belonging to these CPs. This is justified by cache management complexity considerations, as further partitioning of a slice allocated to a CP to be exclusively utilized by the same CP simply incurs additional management complexity. In addition, we show in Appendix~\ref{sec:proof_partition} that partitioning a cache slice into smaller slices and {\em probabilistically routing} requests to content objects of a CP is sub-optimal.
  
As an alternative to CP-oriented cache allocation and partitioning approach, one could adopt a per-object cache allocation and partitioning approach (regardless of the CP or CPs which own the objects). Under such an approach, it is not hard to show that the optimal {\em per-object} cache allocation strategy that maximizes the overall cache hit rate is
equivalent to  the {\em static caching} policy~\cite{Liu98}:  the cache is only allocated to the $C$ most popular objects among all content providers. Alternatively, such a solution can also be obtained using the same CP-oriented, utility maximization cache
allocation framework where only the most popular content from each provider is cached.

\vspace*{2pt}
\noindent
{\bf Utility Functions and Fairness.}
Different utility functions in problems~\eqref{eq:opt} and~\eqref{eq:opt3} yield different partition sizes for content providers. In this sense, each utility function defines a notion of fairness in allocating storage resources to different content providers. The family of $\alpha$-fair utility functions expressed as
\[U(x) = \left\{ \begin{array}{ll}
 \frac{x^{1-\alpha}-1}{1-\alpha} & \alpha \ge 0, \alpha \neq 1; \\
 & \\
 \log{x} & \alpha = 1,
 \end{array} \right. \]
unifies different notions of fairness in resource allocation~\cite{srikant13}. Some choices of $\alpha$ lead to especially interesting utility functions. Table~\ref{tbl:util} gives a brief summary of these functions. We will use these utilities in Section~\ref{sec:simulation} to understand the effect of particular choices for utility functions, and in evaluating our proposed algorithms.

\begin{table}[]
\centering
\caption{$\alpha$-fair utility functions}
\hspace*{-0.5cm}
\begin{tabular}{ | c | c | c | c |}
\hline
$\alpha$ & $U_k(h_k)$ & implication \\
\hline
  0 & $h_k$ & hit rate \\
  1 & $\log{h_k}$ & proportional fairness \\
  2 & $-1/ h_k$ & potential delay \\
  $\infty$ & $\lim_{\alpha \rightarrow \infty} \frac{h_k^{1-\alpha}-1}{1-\alpha}$ & max-min fairness \\
\hline
\end{tabular}
\label{tbl:util}
\end{table}

%% file: online.tex
\section{Online Algorithms}
\label{sec:online}
In the previous section, we formulated cache partitioning as a convex optimization problem. However, it is not feasible to solve the optimization problem offline and then implement the optimal strategy. Moreover, system parameters can change over time. Therefore, we need algorithms that can implement the optimal strategy and adapt to changes in the system by collecting limited information. In this section, we develop such algorithms.

\subsection{Content Providers with Distinct Contents}
The formulation in~\eqref{eq:opt} assumes a hard constraint on the cache capacity. In some circumstances it may be appropriate for the cache manager to increase the available storage at some cost to provide additional resources for the content providers. One way of doing this is to turn cache storage disks on and off based on demand~\cite{Sundarrajan16}. In this case, the cache capacity constraint can be replaced with a penalty function $P(\cdot)$ denoting the cost for the extra cache storage. Here, $P(\cdot)$ is assumed to be convex and increasing. We can now write the utility and cost driven caching formulation as
\begin{align}
\label{eq:hard}
\text{maximize} \quad &\sum_k{U_k(h_k(C_k))} - P(\sum_k{C_k} - C) \notag\\
\text{such that} \quad &C_k \ge 0, \quad k=1,\ldots, K.
\end{align}

Let $W(\mathbf{C})$ denote the objective function in~\eqref{eq:hard} defined as
\[W(\mathbf{C}) = \sum_k{U_k(h_k(C_k))} - P(\sum_k{C_k} - C).\]
A natural approach to obtaining the maximum value for $W(\mathbf{C})$ is to use a gradient ascent algorithm. The basic idea behind a gradient ascent algorithm is to move the variables $C_k$ in the direction of the gradient,
\begin{align*}
\frac{\partial W}{\partial C_k} &= \frac{\partial U_k}{\partial C_k} - P'(\sum_k{C_k} - C), \\
&= U'_k(h_k)\frac{\partial h_k}{\partial C_k} - P'(\sum_k{C_k} - C).
\end{align*}
Note that since $h_k$ is an increasing function of $C_k$, moving $C_k$ in the direction of the gradient also moves $h_k$ in that direction.

By gradient ascent, partition sizes should be updated according to
\[C_k \gets \max{\Big\{0, C_k + \gamma_k\Big[ U'_k(h_k)\frac{\partial h_k}{\partial C_k} - P'(\sum_k{C_k} - C) \Big]\Big\}},\]
where $\gamma_k$ is a step-size parameter.
\begin{theorem}
\label{THM:GA}
The above gradient ascent algorithm converges to the optimal solution.
\end{theorem}
\begin{proof}
Let $\mathbf{C}^*$ denote the optimal solution to (4). We show in Appendix~\ref{sec:lyapunov} that $W(\mathbf{C}^*) - W(\mathbf{C})$ is a Lyapunov function, and the above algorithm converges to the optimal solution.
\end{proof}

\subsubsection{Algorithm Implementation}
\label{sec:impl}
In implementing the gradient ascent algorithm, we restrict ourselves to the case where the total cache size is $C$. Defining ${\eta \triangleq P'(0)}$, we can re-write the gradient ascent algorithm as
\[C_k \gets \max{\Big\{0, C_k + \gamma_k\Big[ U'_k(h_k)\frac{\partial h_k}{\partial C_k} - \eta \Big]\Big\}}.\]
In order to update $C_k$ then, the cache manager needs to estimate $\frac{\partial U_k}{\partial C_k} = U'_k(h_k)\frac{\partial h_k}{\partial C_k}$ by gathering hit rate information for each content provider. Instead of computing $U'_k(h_k)$ and $\partial h_k/\partial C_k$ separately, however, we suggest using
\[\frac{\partial U_k}{\partial C_k} \approx \frac{\Delta U_k}{\Delta C_k} = \frac{U_k(h_k^t) - U_k(h_k^{t-1})}{C_k^t - C_k^{t-1}},\]
where the superscripts $t$ and $t-1$ denote the iteration steps. We then use $\frac{\Delta U_k}{\Delta C_k}$ as an estimate of $U'_k(h_k)\frac{\partial h_k}{\partial C_k}$ to determine the value of $C_k$ at the next iteration.

Moreover, since we impose the constraint that $\sum_k{C_k} = C$, we let $\eta$ take the mean of the $\frac{\Delta U_k}{\Delta C_k}$ values. The algorithm reaches a stable point once the $\frac{\Delta U_k}{\Delta C_k}$s are equal or very close to each other. Algorithm~\ref{alg:online} shows the rules for updating the partition sizes.

\begin{algorithm}[t]
\caption{Online algorithm for updating the partition sizes.}
\label{alg:online}
\begin{algorithmic}[1]
	\Statex
		\State Start with an initial partitioning ${\mathbf{C}^0\gets (C_1 \cup C_2 \cup \ldots \cup C_K)}$.
		\State Estimate hit rates for each partition by counting the number of hit requests ${\mathbf{h}^0\gets (h_1, \ldots, h_K)}$.
		\State Make arbitrary changes to the partition sizes $\mathbf{\Delta}^0$, such that $\sum_{k}{\mathbf{\Delta}^0_k}=0$, ${\mathbf{C}^1\gets \mathbf{C}^0 + \mathbf{\Delta}^0}$.
		\State Estimate the hit rates $\mathbf{h}^1$ for the new partition sizes.
		\State $t\gets 1$.
		\State $\delta^t_k\gets \Big(U_k(h_k^t) - U_k(h_k^{t-1})\Big)/\Big(C_k^t - C_k^{t-1}\Big)$.
		\State $\eta^t \gets \Big(\sum_{k}{\delta^t_k}\Big)/K$.
		\If{$\max_{k}{\{\delta^t_k - \eta^t\}} > \epsilon$}
			\State $\mathbf{\Delta}_k^t = \gamma(\delta^t_k - \eta^t)$.
			\State $\mathbf{C}^{t+1}\gets \mathbf{C}^t + \mathbf{\Delta}^t$.
			\State Estimate hit rates $\mathbf{h}^{t+1}$.
			\State $t\gets t+1$.
			\State \textbf{goto} 6.
		\EndIf
\end{algorithmic}
\end{algorithm}

\subsection{Content Providers with Common Content}
We now focus on the case where some contents can be served by multiple content providers.
Algorithm~\ref{alg:partition} computes the optimal number of partitions for this case. Let $\mathcal{P}$ and ${\mathbf{C} = (C_1,\ldots, C_{|\mathcal{P}|})}$ denote the set of partitions and the vector of partition sizes, respectively. The hit rate for content provider $k$ can be written as
\[h_k(\mathbf{C}) = \sum_{p=1}^{|\mathcal{P}|}{\sum_{i\in V_p}{\lambda_{ik}(1 - e^{-\lambda_iT_p})}},\]
where $V_p$ denotes the set of files requested from partition $p$, and $\lambda_i$ denotes the aggregate request rate at partition $p$ for file $i$.

Similar to~\eqref{eq:hard}, we consider a penalty function for violating the cache size constraint and rewrite the optimization problem in~\eqref{eq:opt3} as
\begin{align}
\label{eq:hard2}
\text{maximize} \quad &\sum_k{U_k(h_k(\mathbf{C}))} - P(\sum_p{C_p} - C) \notag\\
\text{such that} \quad &C_p \ge 0, \quad p=1,\ldots, |\mathcal{P}|.
\end{align}
Let $W(\mathbf{C})$ denote the objective function in the above problem. Taking the derivative of $W$ with respect to $C_p$ yields
\[\frac{\partial W}{\partial C_p} = \sum_k{U'_k(h_k)\frac{\partial h_k}{\partial C_p}} - P'(\sum_p{C_p} - C).\]

Following a similar argument as in the previous section, we can show that a gradient ascent algorithm converges to the optimal solution.

\subsubsection{Algorithm Implementation}
The implementation of the gradient ascent algorithm in this case is similar to the one in Section~\ref{sec:impl}. However, we need to keep track of hit rates for content provider $k$ from all partitions that store its files. This can be done by counting the number of hit requests for each content provider and each partition through a $K\times |\mathcal{P}|$ matrix, as shown in Algorithm~\ref{alg:online_shared}. Also, we propose estimating $\partial W/\partial C_p$ as
\[\frac{\partial W}{\partial C_p} \approx \sum_k{U'_k(h_k)\frac{\Delta h_{kp}}{\Delta C_p}},\]
where $\Delta h_{kp}$ denotes the change in aggregate hit rate for content provider $k$ from partition $p$ resulted from changing the size of partition $p$ by $\Delta C_p$.

\begin{algorithm}[t]
\caption{Online algorithm for updating the partition sizes.}
\label{alg:online_shared}
\begin{algorithmic}[1]
	\Statex
		\State Compute the number of partitions $P$ using Algorithm~\ref{alg:partition}.
		\State Start with an initial partitioning ${\mathbf{C}^0\gets (C_1 \cup C_2 \cup \ldots \cup C_P)}$.
		\State Estimate hit rates for each provider/partition pair ${\mathbf{H}^0\gets \left( \begin{array}{ccc}
h_{11} & \ldots & h_{1P} \\
\vdots & \ddots & \vdots \\
h_{K1} & \ldots & h_{KP} \end{array} \right)}$.
		\State Make arbitrary changes to the partition sizes $\mathbf{\Delta}^0$, such that $\sum_{p}{\mathbf{\Delta}^0_p}=0$, ${\mathbf{C}^1\gets \mathbf{C}^0 + \mathbf{\Delta}^0}$.
		\State Estimate the hit rates $\mathbf{H}^1$ for the new partition sizes.
		\State $t\gets 1$.
		\State $\delta^t_p\gets \sum_k{U'_k(h_k^{t-1})(h_{kp}^{t} - h_{kp}^{t-1})/(C_p^t - C_p^{t-1})}$.
		\State $\eta^t \gets \Big(\sum_{p}{\delta^t_p}\Big)/P$.
		\If{$\max_{p}{\{\delta^t_p - \eta^t\}} > \epsilon$}
			\State $\mathbf{\Delta}_p^t = \gamma(\delta^t_p - \eta^t)$.
			\State $\mathbf{C}^{t+1}\gets \mathbf{C}^t + \mathbf{\Delta}^t$.
			\State Estimate hit rates $\mathbf{H}^{t+1}$.
			\State $t\gets t+1$.
			\State \textbf{goto} 6.
		\EndIf
\end{algorithmic}
\end{algorithm}

%% file: simulation.tex
\section{Evaluation}
\label{sec:simulation}
In this section, we perform numerical simulations, first to understand the efficacy of cache partitioning on the utility observed by content providers, and second to evaluate the performance of our proposed online algorithms.

For our base case, we consider a cache with capacity ${C=10^4}$. Each partition uses LRU as the cache management policy. We consider two content providers that serve $n_1 = 10^4$ and $n_2 = 2\times 10^4$ contents. Content popularities for the two providers follow Zipf distributions, \ie\ $p_i \propto 1/i^z$, with parameters $z_1 = 0.6$ and $z_2 = 0.8$, respectively. Requests for the files from the two content providers arrive as Poisson processes with aggregate rates $\lambda_1 = 15$ and ${\lambda_2=10}$. The utilities of the two content providers are $U_1(h_1) = w_1\log{h_1}$ and $U_2(h_2) = h_2$. Unless otherwise specified, we let $w_1 = 1$ so that the two content providers are equally important to the 
service provider.

We consider two scenarios here. In the first scenario, the two content providers serve completely separate files. In the second scenario, files ${S_0=\{1,4,7,\ldots,10^4\}}$, are served by both providers. For each scenario the appropriate optimization formulation is chosen.

\subsection{Cache Partitioning}

\begin{figure}[]
\centering
\includegraphics[scale=0.4]{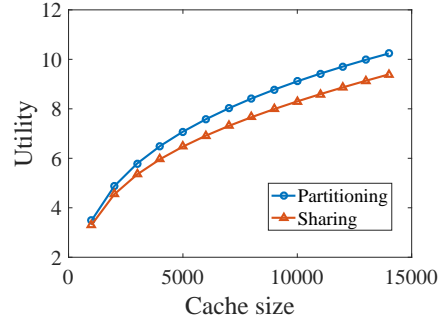} 
 \caption{Efficacy of cache partitioning when content providers serve distinct files.}
    \centering\label{fig:up}
\end{figure}

\begin{figure}[]
\centering
 \begin{subfigure}[b]{0.5\linewidth}
  	\centering\includegraphics[scale=0.3]{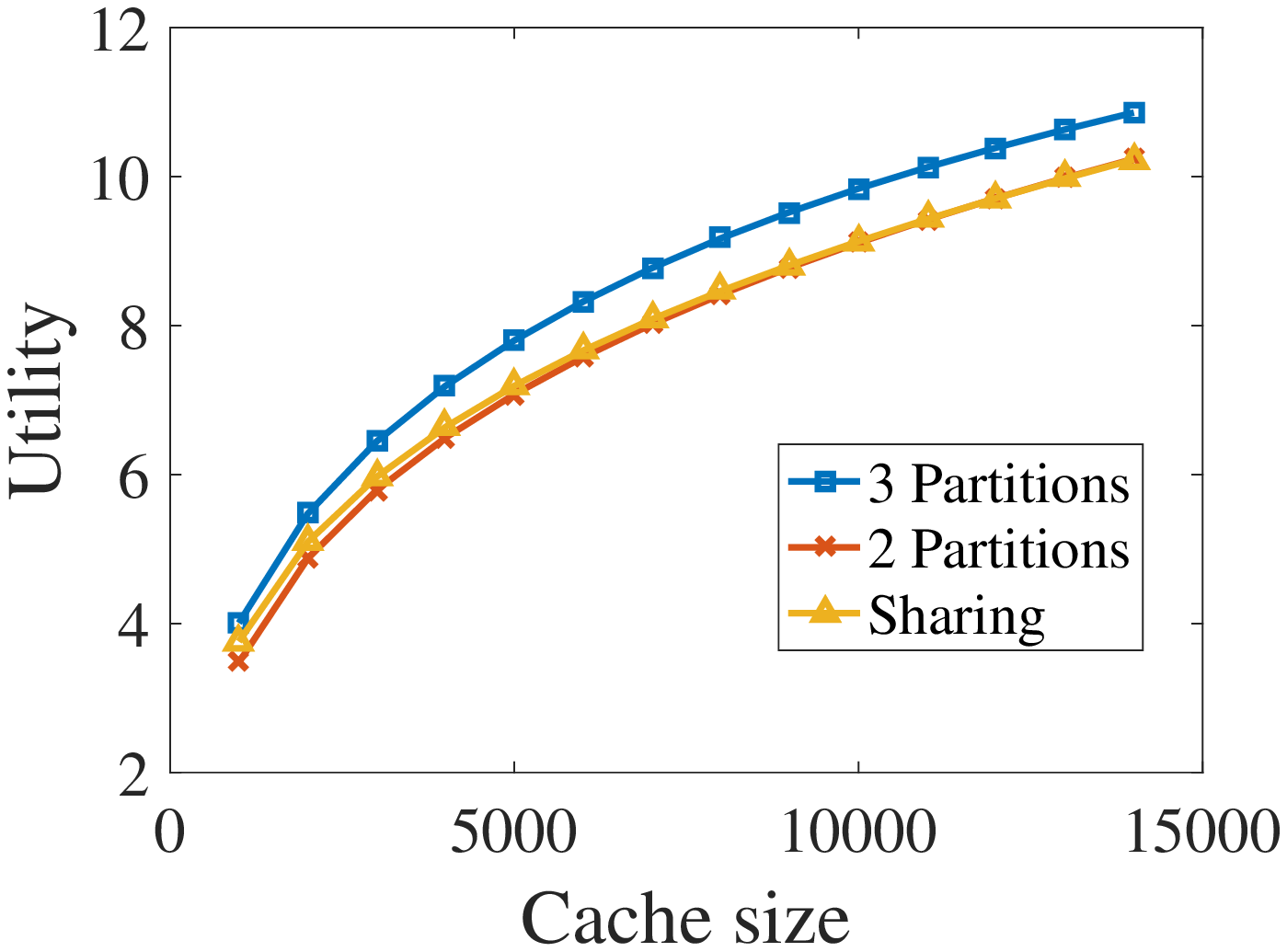} 
  	\caption{}
 \end{subfigure}%
 \begin{subfigure}[b]{0.5\linewidth}
  	\centering\includegraphics[scale=0.3]{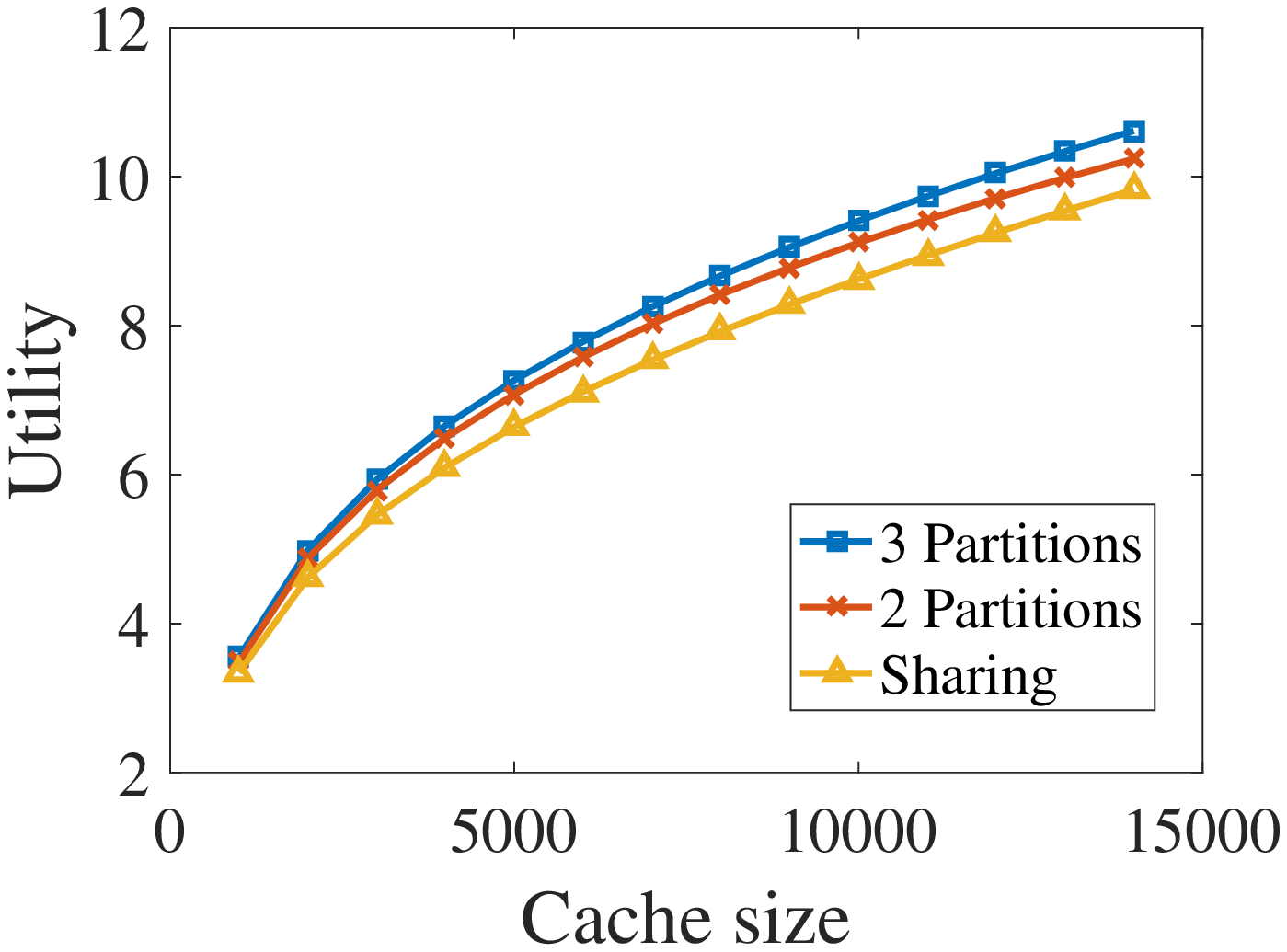}
  	\caption{}
 \end{subfigure}%
 \caption{Efficacy of cache partitioning when some content is served by both content providers. Request rates for the common contents from the two content providers are set to be (a) similar, and (b) dissimilar.}
    \centering\label{fig:up_shared}
\end{figure}

\begin{figure*}[t]
\centering
\begin{subfigure}[b]{\linewidth}
  	\centering\includegraphics[scale=0.3]{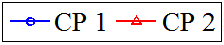}
 \end{subfigure}
 \begin{subfigure}[b]{0.25\linewidth}
  	\centering\includegraphics[scale=0.3]{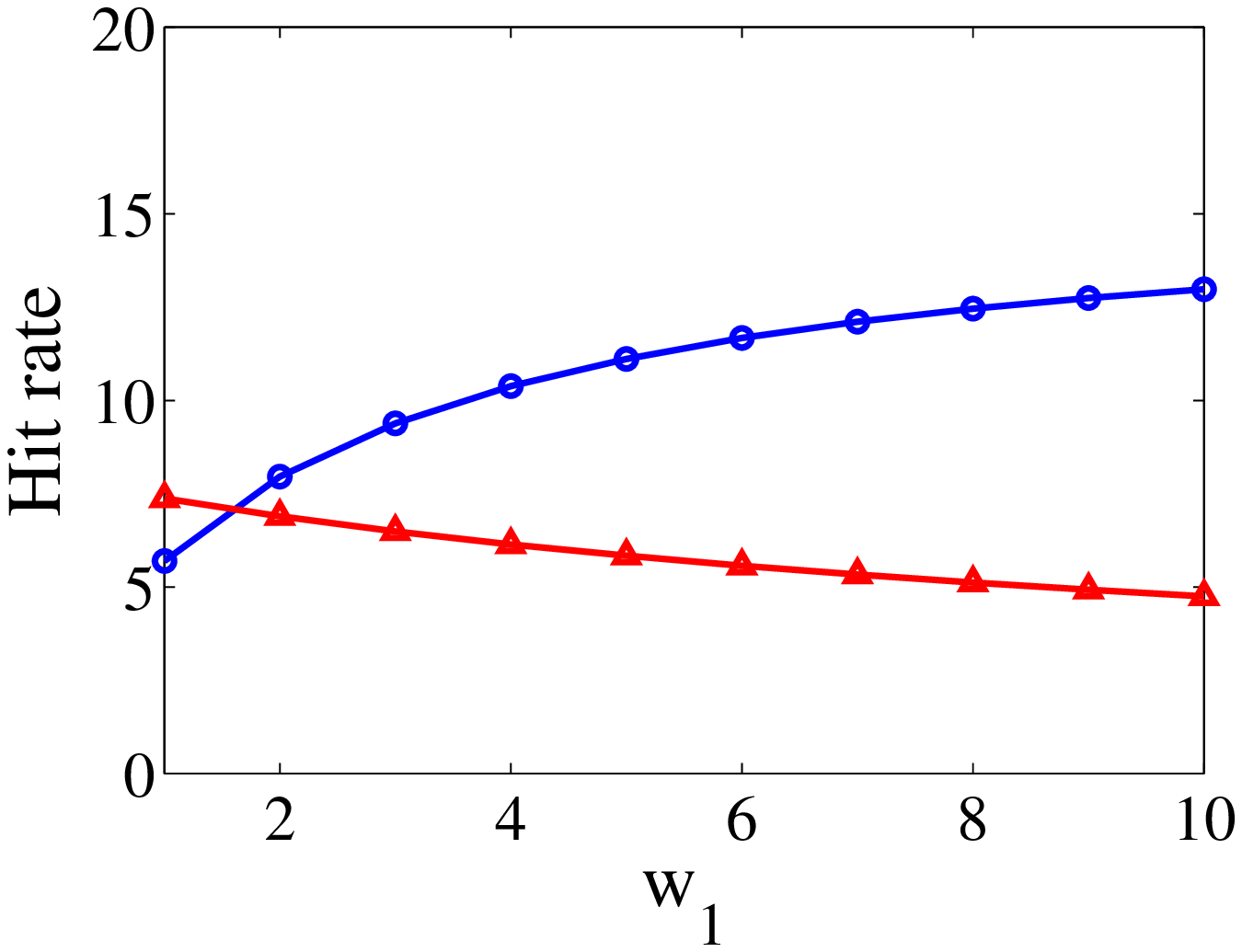}
 \end{subfigure}%
 \begin{subfigure}[b]{0.25\linewidth}
  	\centering\includegraphics[scale=0.3]{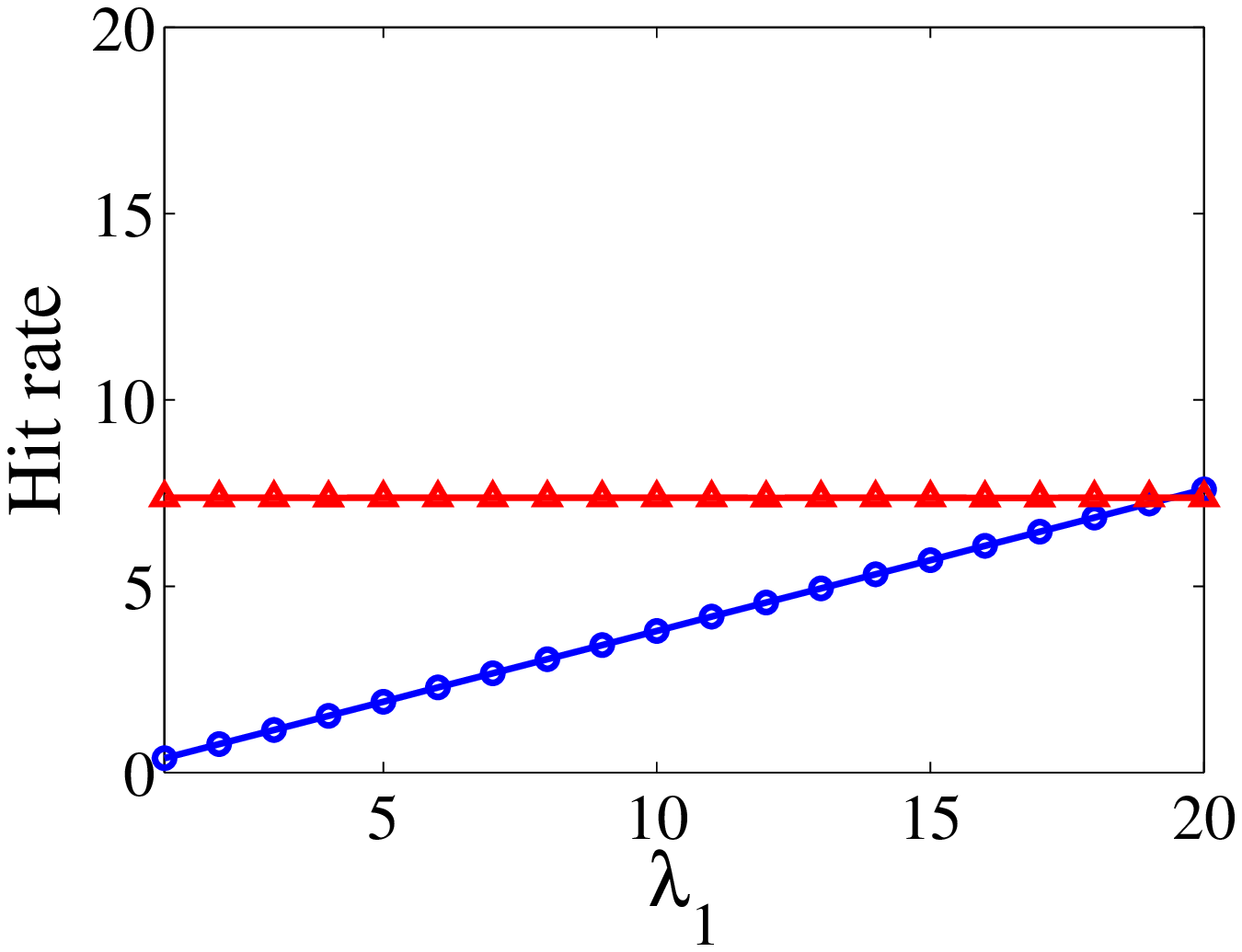}
 \end{subfigure}%
 \begin{subfigure}[b]{0.25\linewidth}
  	\centering\includegraphics[scale=0.3]{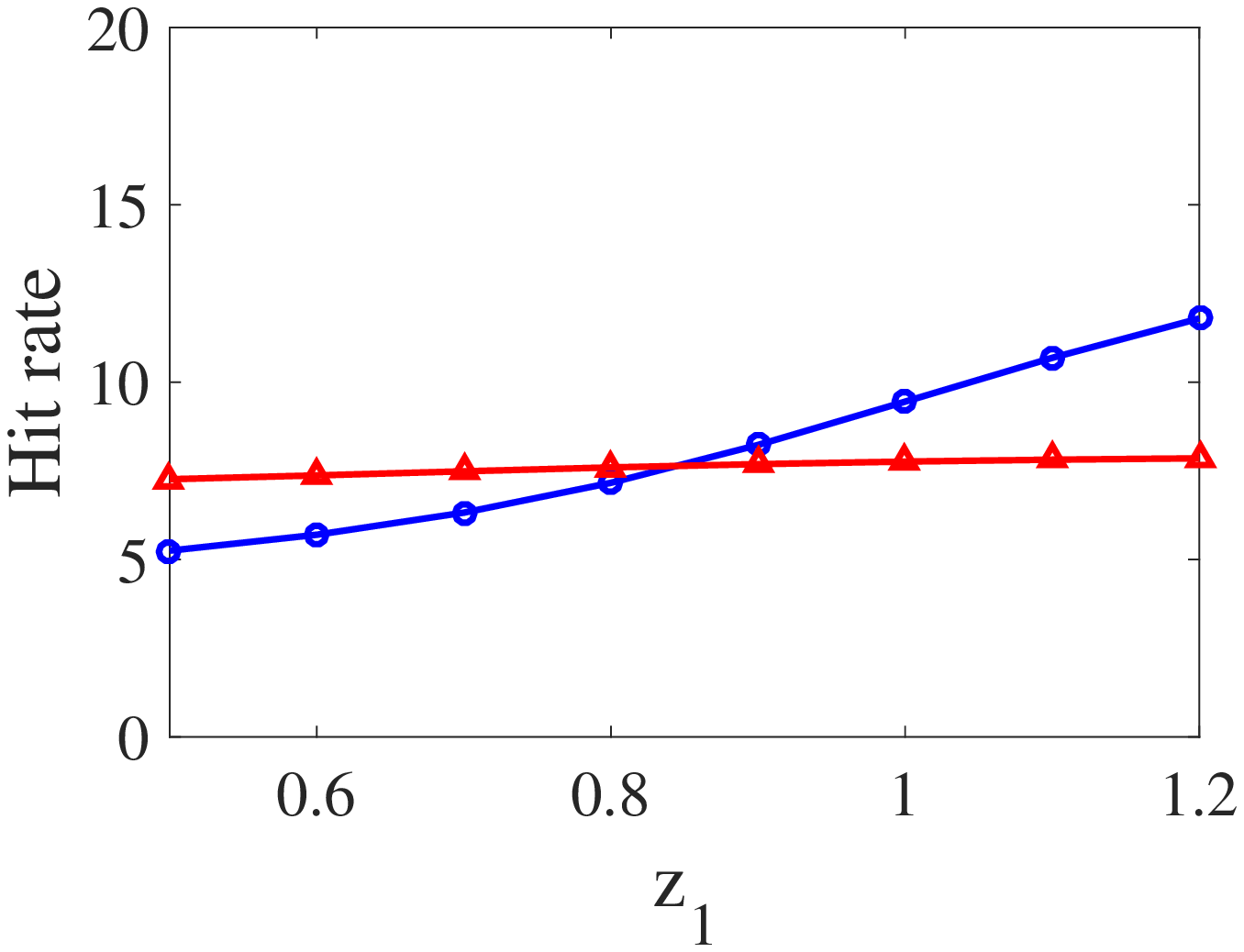}
 \end{subfigure}%
  \begin{subfigure}[b]{0.25\linewidth}
  	\centering\includegraphics[scale=0.3]{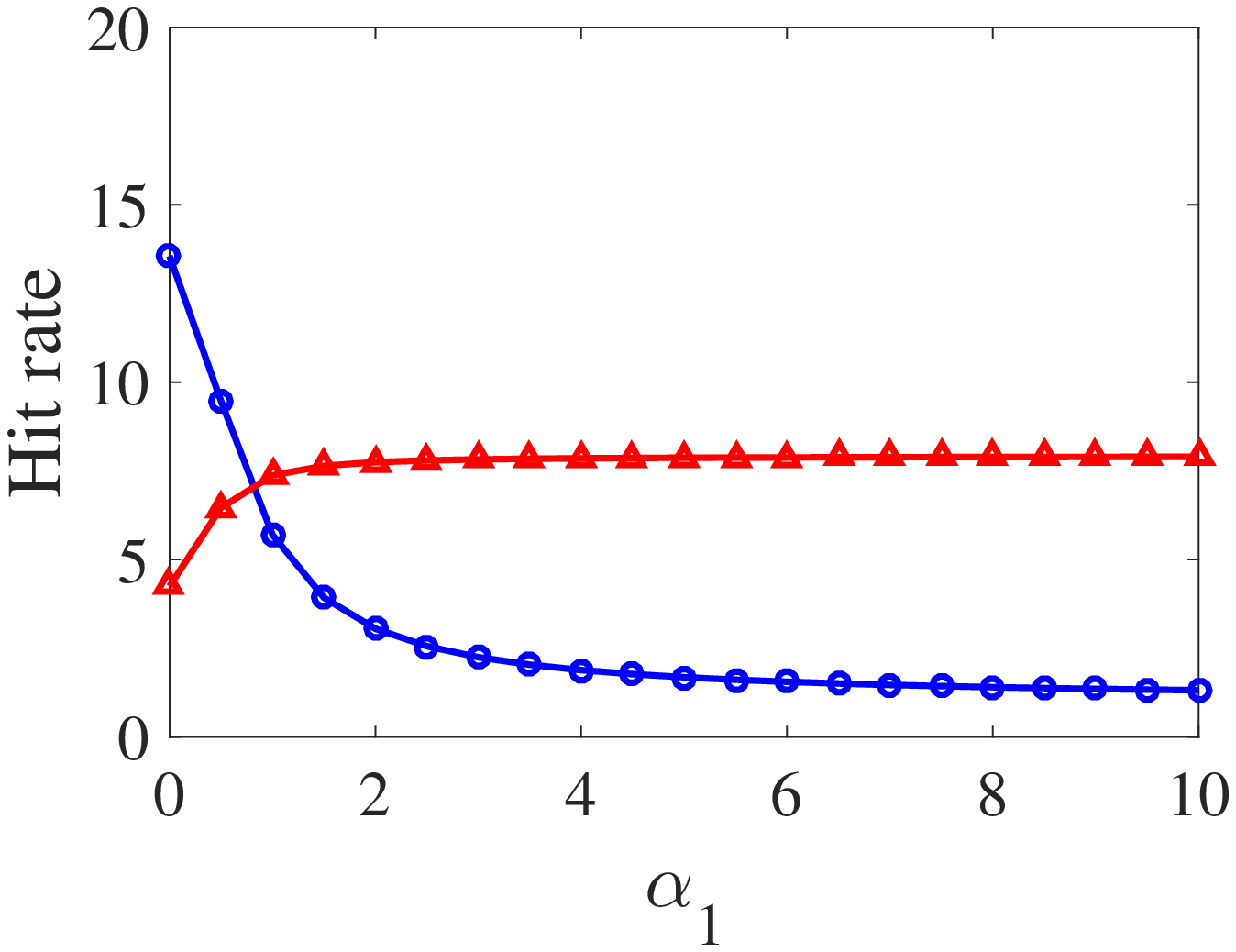}
 \end{subfigure}
  \begin{subfigure}[b]{0.25\linewidth}
  	\centering\includegraphics[scale=0.3]{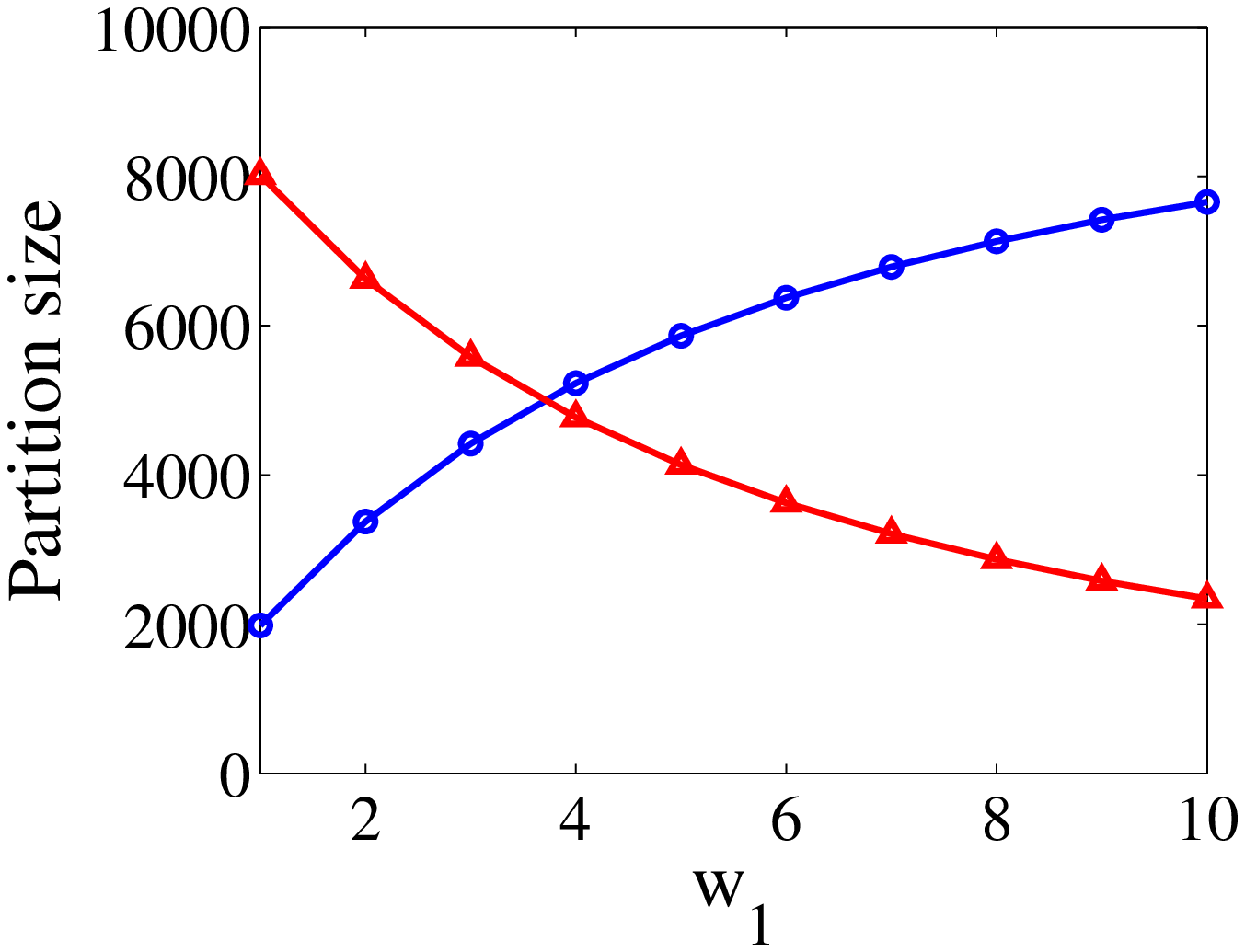}
 \end{subfigure}%
 \begin{subfigure}[b]{0.25\linewidth}
  	\centering\includegraphics[scale=0.3]{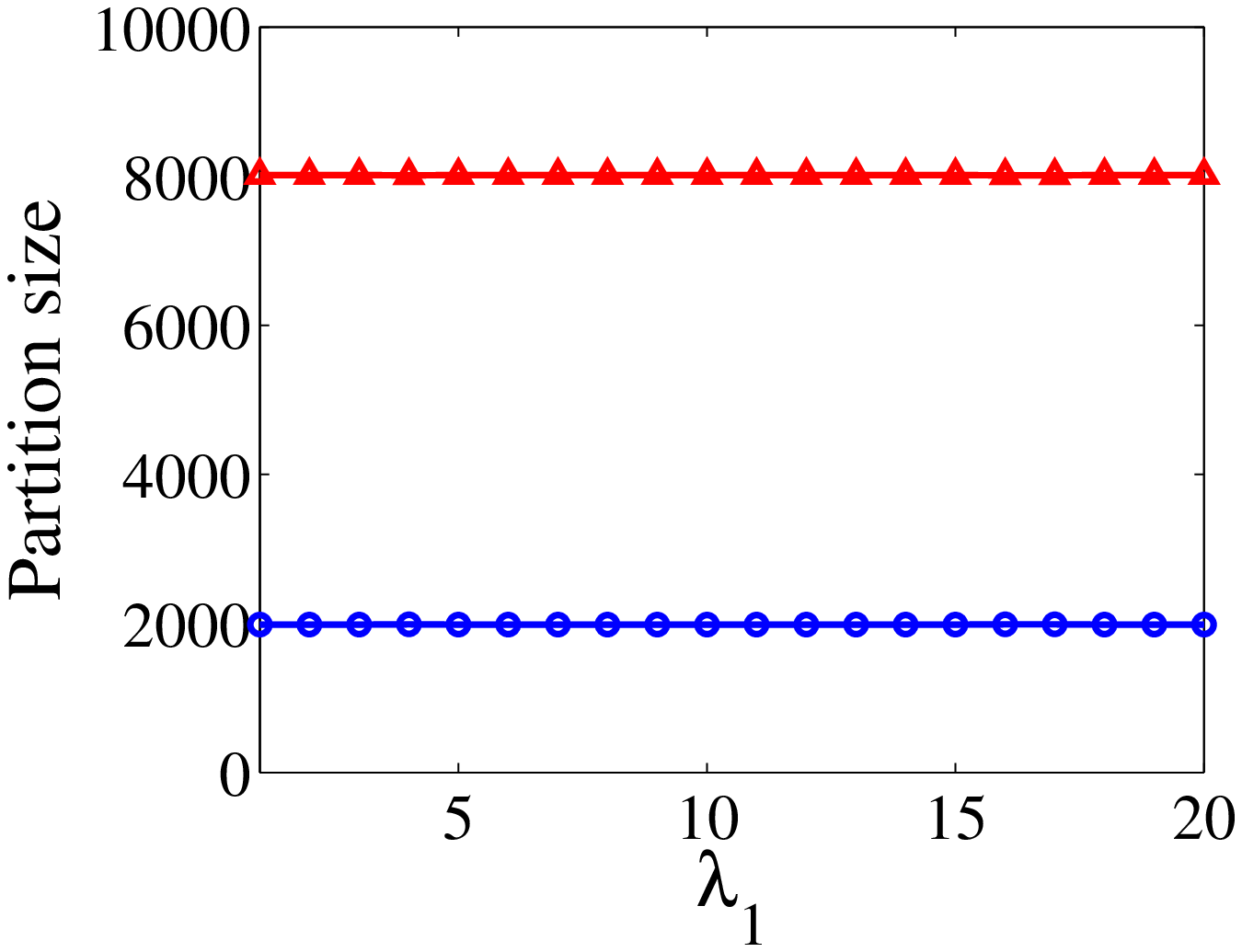}
 \end{subfigure}%
 \begin{subfigure}[b]{0.25\linewidth}
  	\centering\includegraphics[scale=0.3]{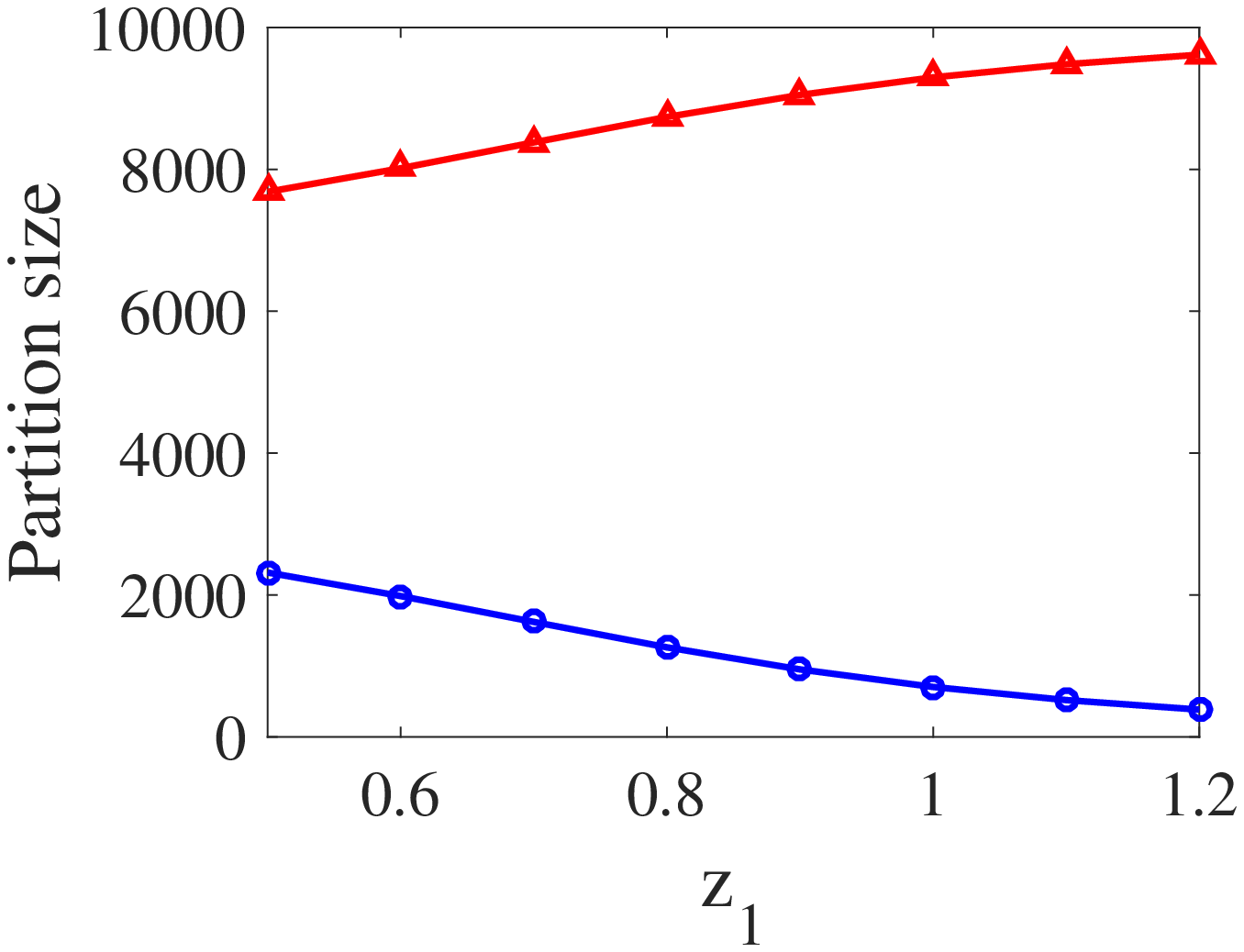}
 \end{subfigure}%
 \begin{subfigure}[b]{0.25\linewidth}
  	\centering\includegraphics[scale=0.3]{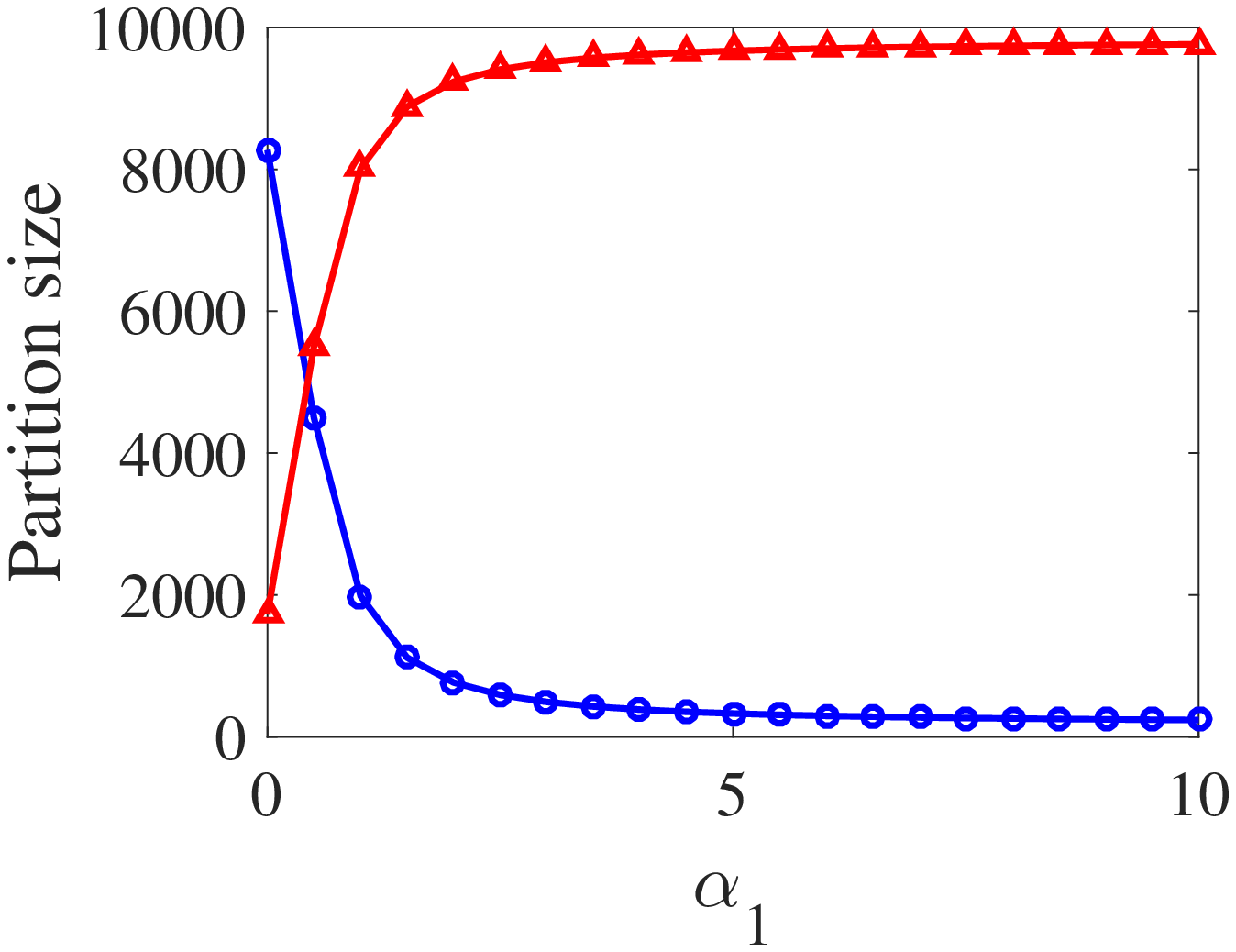}
 \end{subfigure}
 \caption{Effect of the parameters on hit rates and partition sizes when content providers serve distinct files.}
    \centering\label{fig:bwlama_hitrates}
\end{figure*}

\begin{figure*}[t]
\centering
\begin{subfigure}[b]{\linewidth}
  	\centering\includegraphics[scale=0.3]{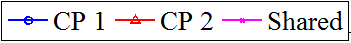}
 \end{subfigure}
 \begin{subfigure}[b]{0.25\linewidth}
  	\centering\includegraphics[scale=0.3]{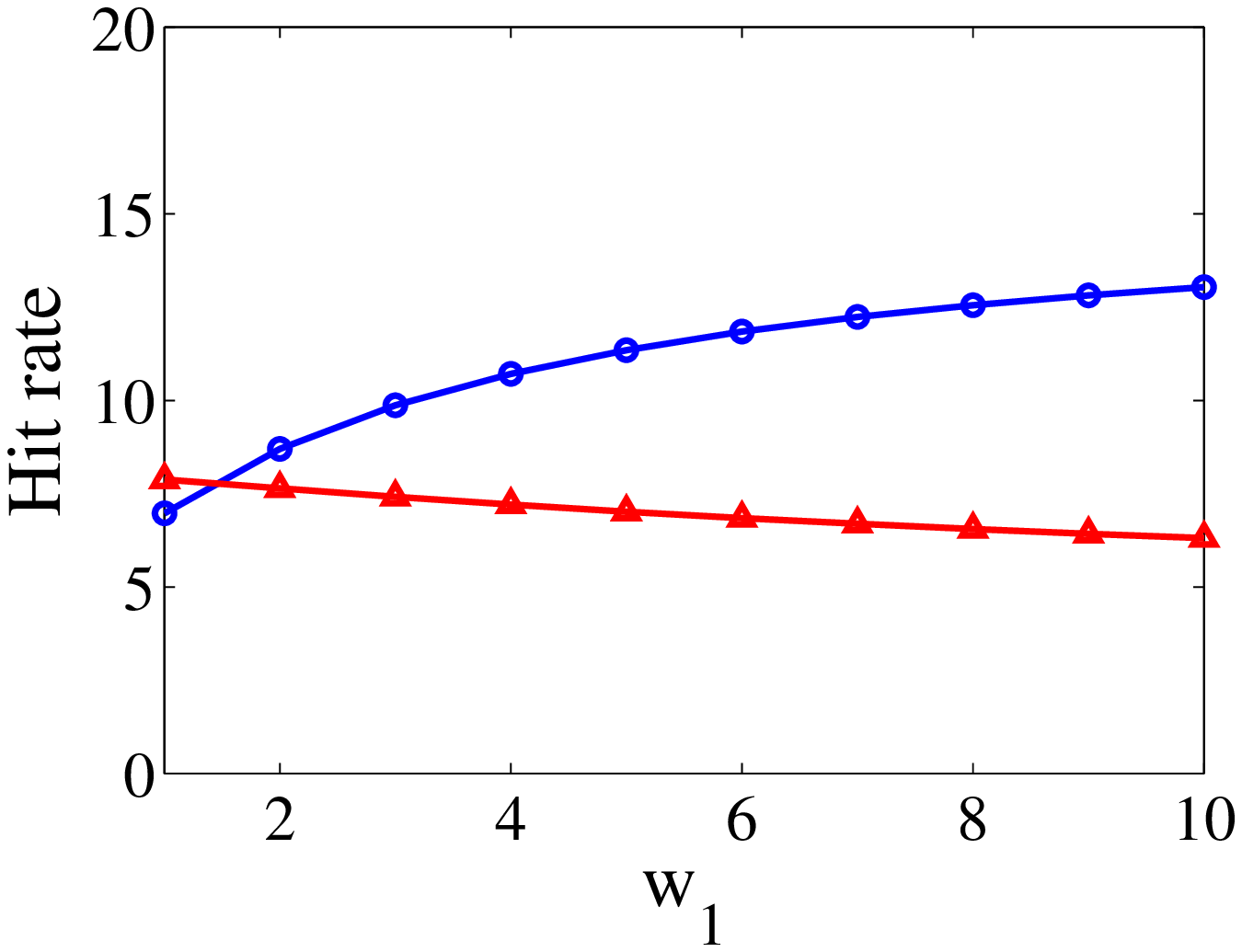}
 \end{subfigure}%
 \begin{subfigure}[b]{0.25\linewidth}
  	\centering\includegraphics[scale=0.3]{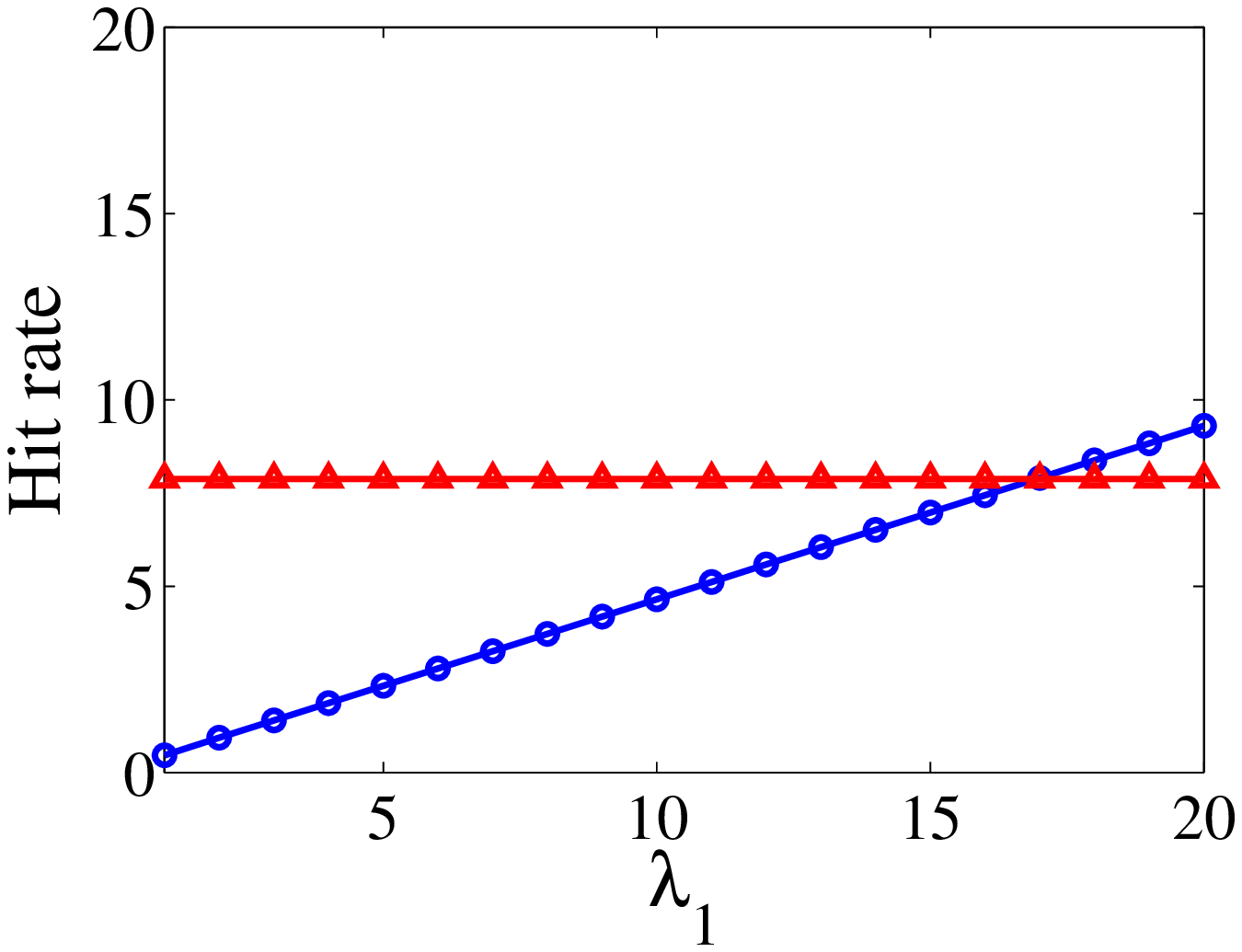}
 \end{subfigure}%
 \begin{subfigure}[b]{0.25\linewidth}
  	\centering\includegraphics[scale=0.3]{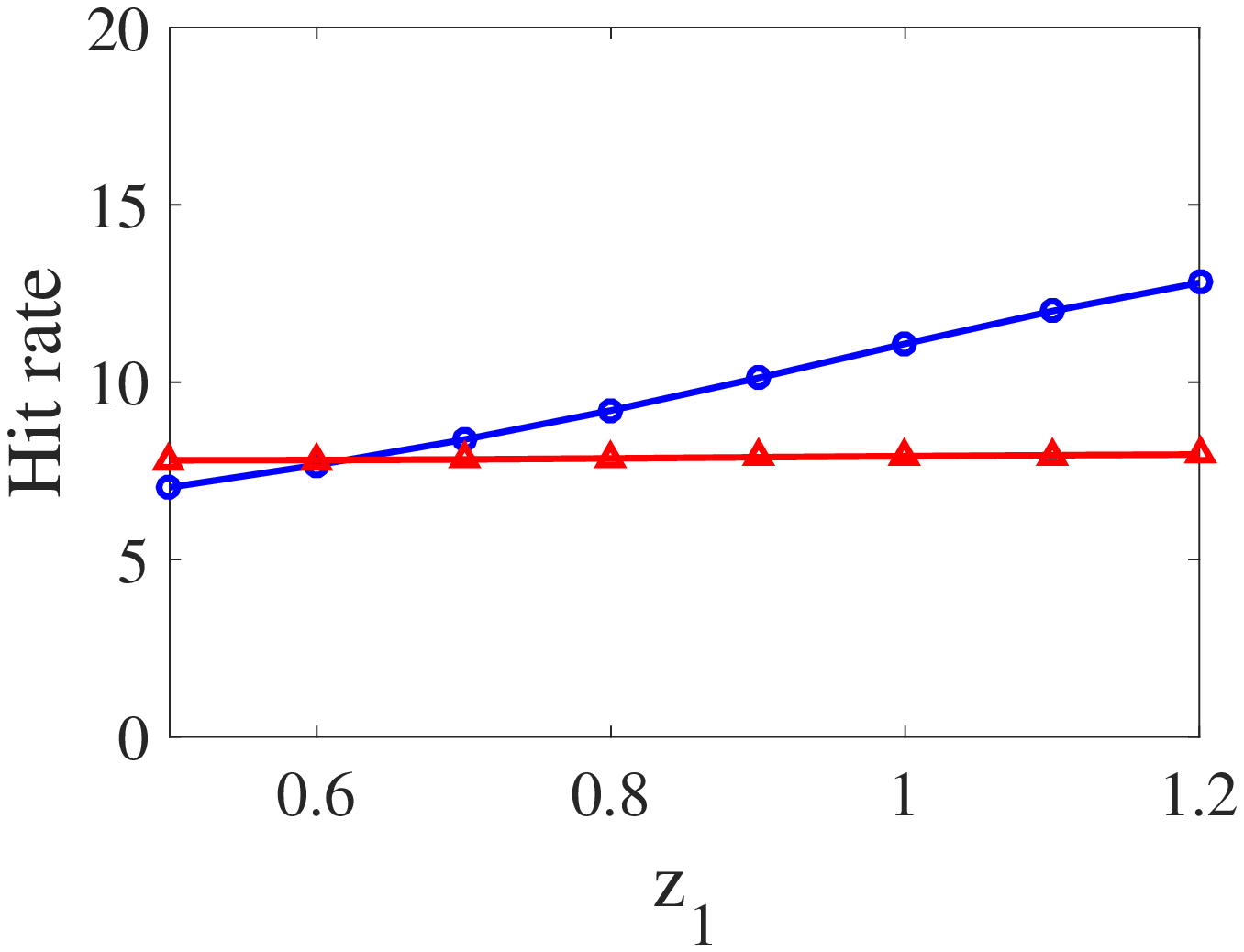}
 \end{subfigure}%
 \begin{subfigure}[b]{0.25\linewidth}
  	\centering\includegraphics[scale=0.3]{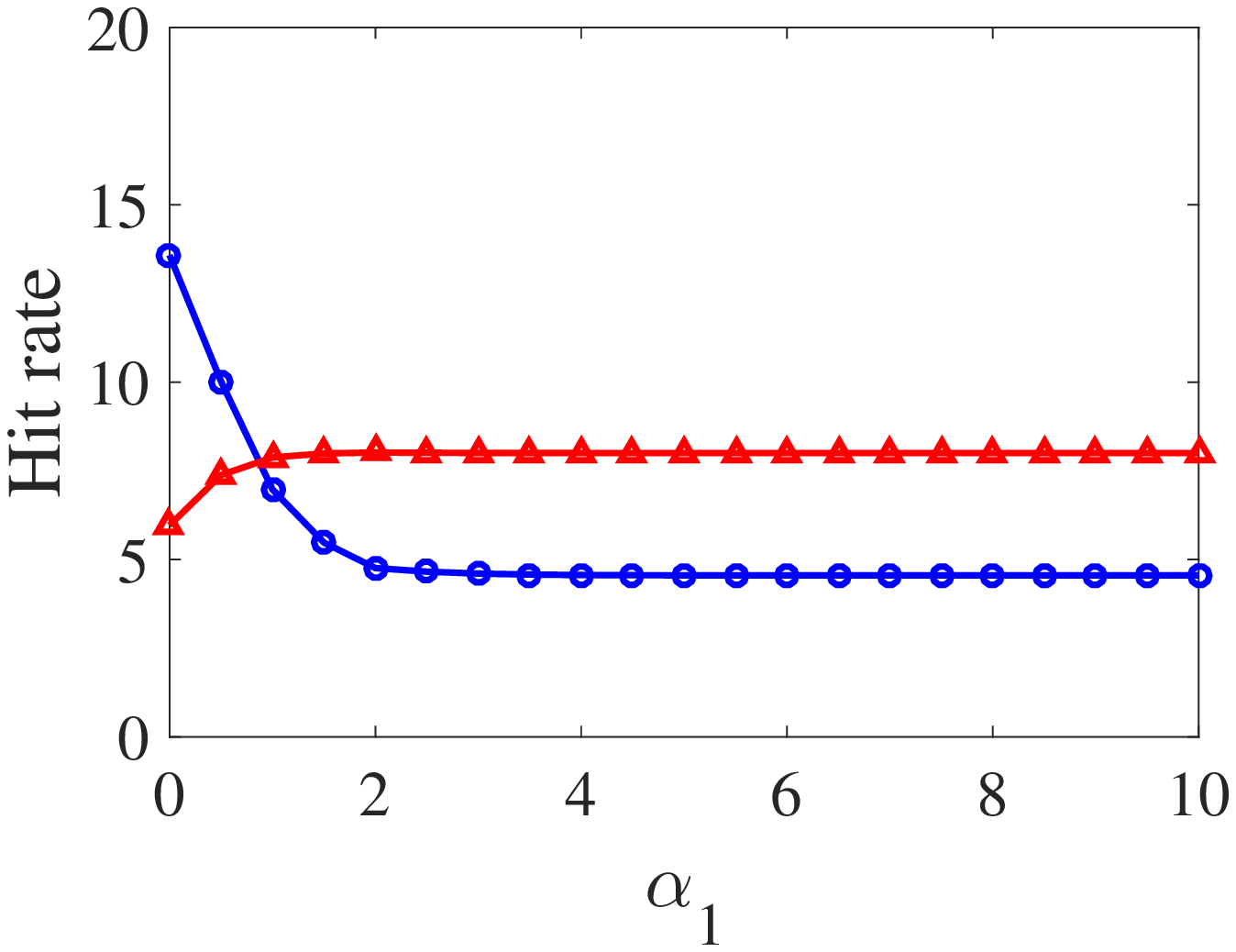}
 \end{subfigure}
  \begin{subfigure}[b]{0.25\linewidth}
  	\centering\includegraphics[scale=0.3]{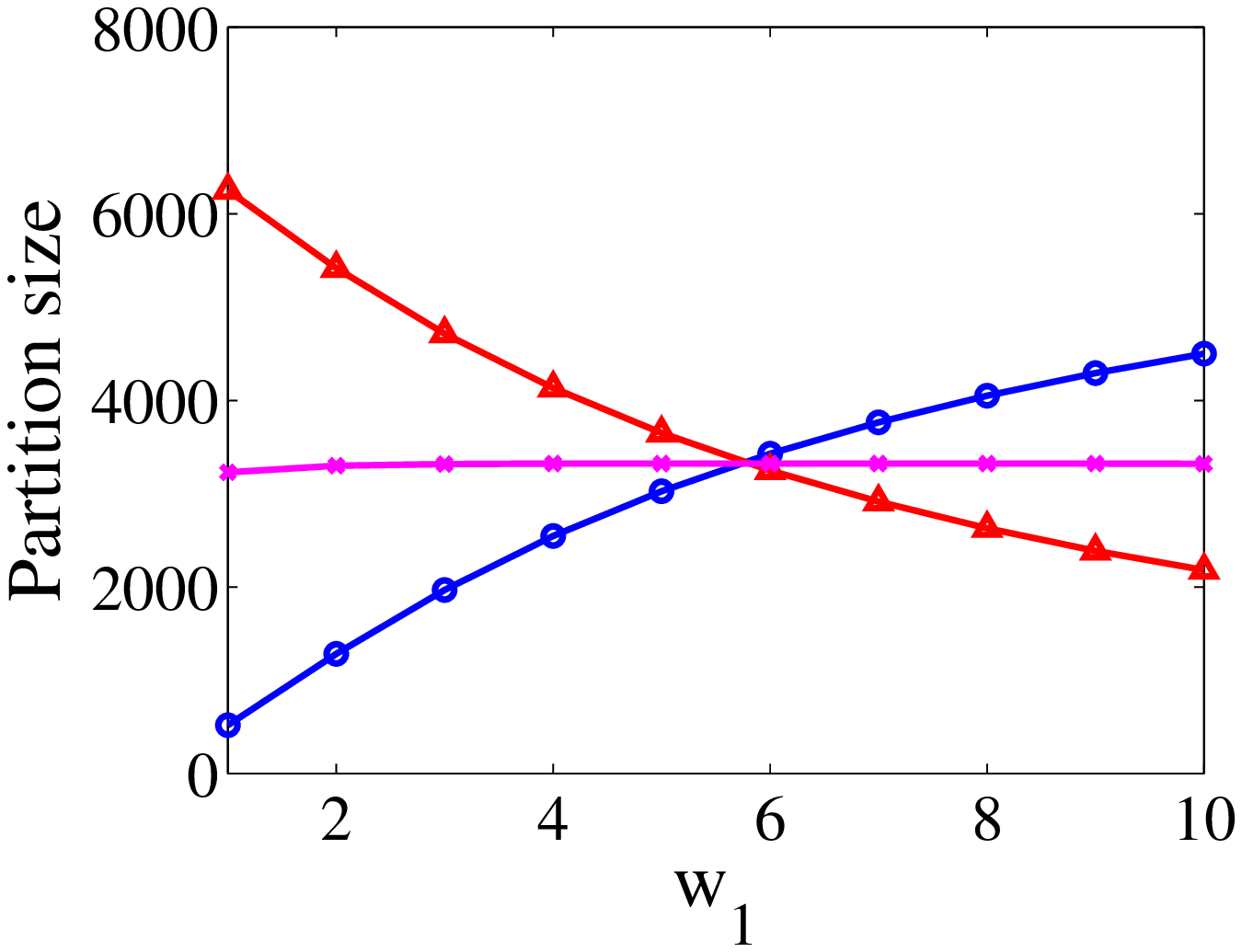}
 \end{subfigure}%
 \begin{subfigure}[b]{0.25\linewidth}
  	\centering\includegraphics[scale=0.3]{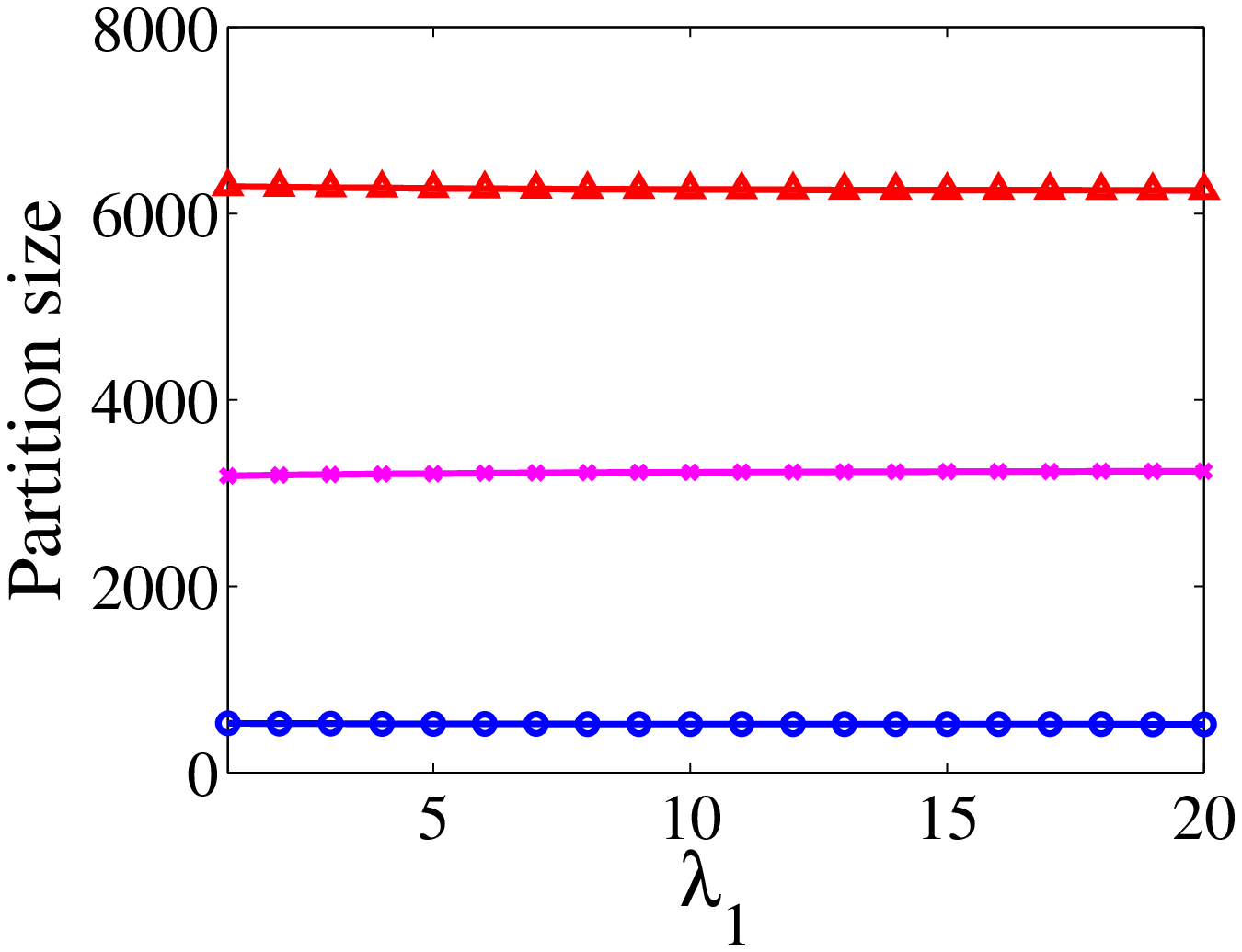}
 \end{subfigure}%
 \begin{subfigure}[b]{0.25\linewidth}
  	\centering\includegraphics[scale=0.3]{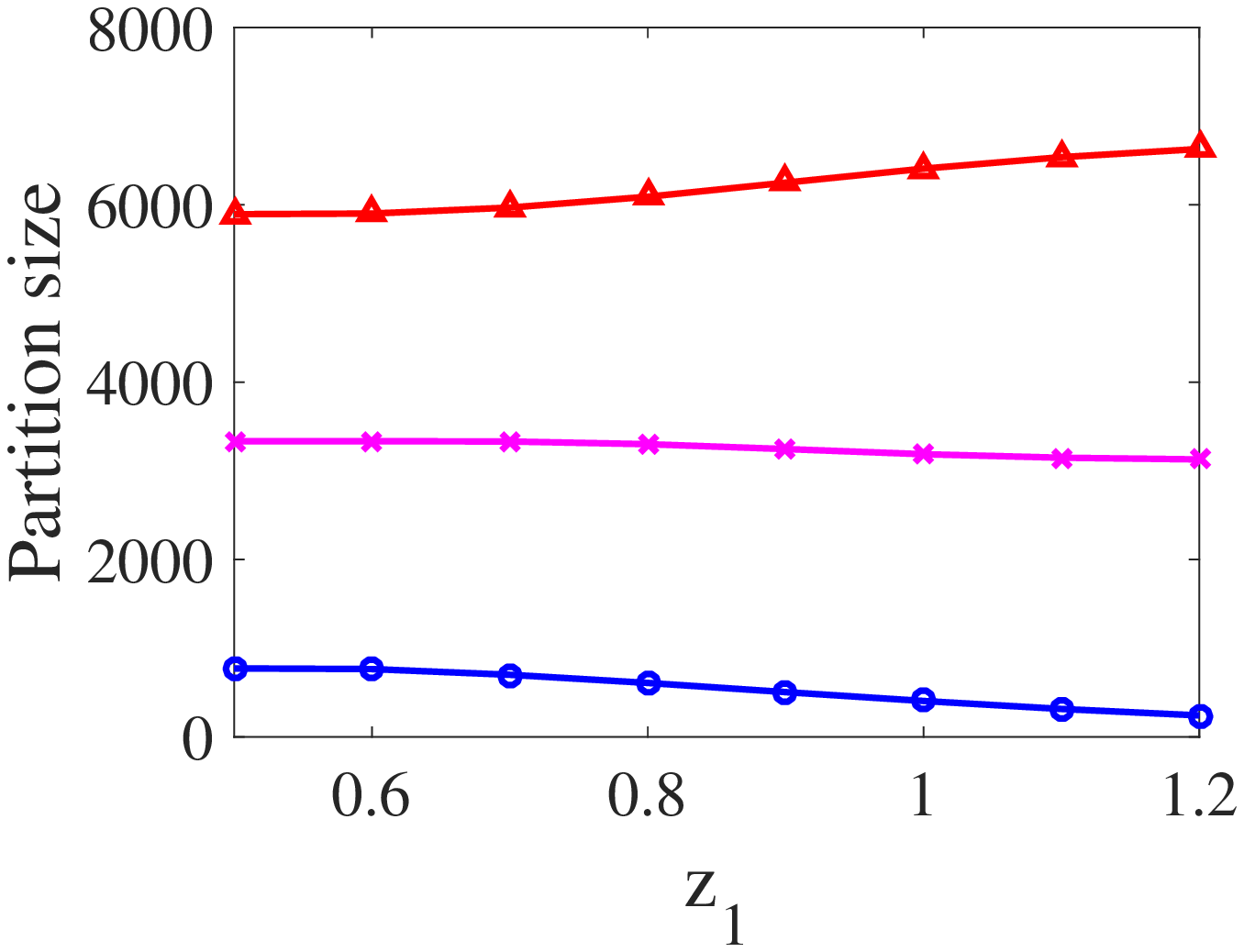}
 \end{subfigure}%
 \begin{subfigure}[b]{0.25\linewidth}
  	\centering\includegraphics[scale=0.3]{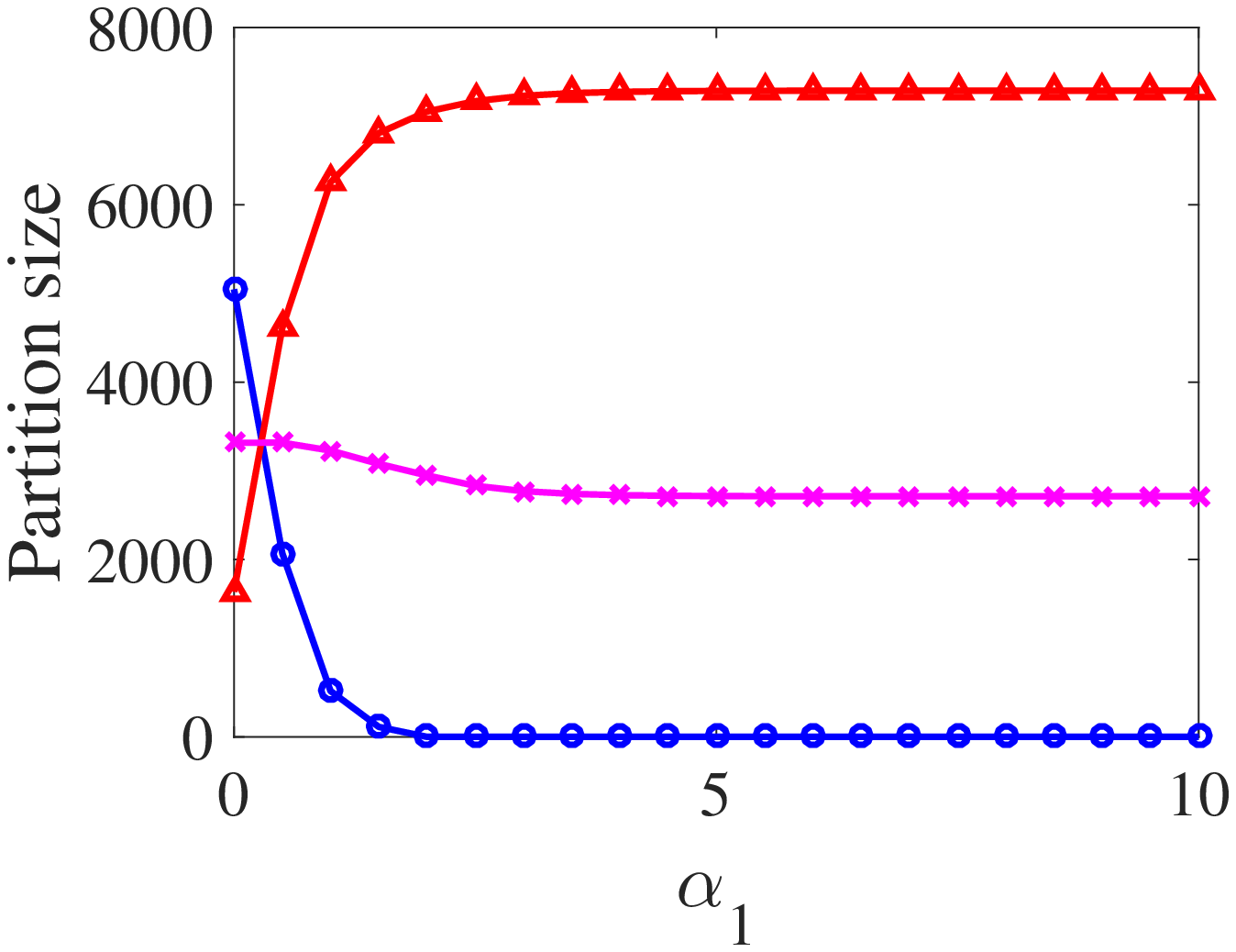}
 \end{subfigure}
 \caption{Effect of the parameters on hit rates and partition sizes when some content is served by both content providers.}
    \centering\label{fig:bwlama_hitrates_shared}
\end{figure*}

To understand the efficacy of cache partitioning, we first look at solutions of optimization problems~\eqref{eq:opt} and~\eqref{eq:opt3}. Here, we measure the gain in total utility through partitioning the cache by computing the utility obtained by sharing the cache between the content providers and the utility obtained by partitioning the cache. Figure~\ref{fig:up} shows the utility gain when content providers serve distinct files. In this example, the aggregate utility increases by $10\%$ from partitioning the cache.

Figure~\ref{fig:up_shared} shows the utilities for the case when files in ${S_0 = \{1, 4, 7, \ldots, 10^4\}}$ are served by both content providers. Two cases are considered here: a)~request rates for the common content are similar for two content providers. This is done by letting ${p_{k,1} > p_{k,2} > \ldots > p_{k, n_0}}$ for both providers. b)~Requests rates from the two content providers for the common files are set to be dissimilar. This is done by setting the file popularities for the second CP as ${p_{2,1} < p_{2,2} < \ldots < p_{2, n_0}}$. In both cases partitioning the cache into three slices shows the best performance.

We next look at the effect of various parameters on cache partitioning, when CPs serve distinct contents and when they serve some common content with similar popularities. We fix the parameters of the second content provider, and study the effect of changing weight parameter $w_1$ and aggregate request rate $\lambda_1$ of the first content provider. We also change the Zipfian file popularity distribution parameter $z_1$. To study the effect of the utility function, we take it to be the $\alpha$--fair utility function and vary $\alpha$ for the first content provider, $\alpha_1$.

Figure~\ref{fig:bwlama_hitrates} shows how hit rates and partition sizes of the two content providers vary as functions of $w_1$, $\lambda_1$, $z_1$ and $\alpha_1$. As expected, by increasing the weight $w_1$, content provider one gets a larger share of the cache, and hence a higher hit rate. Increasing $\lambda_1$ has no effect on the partition sizes. This is because the first content provider uses the $\log$ utility function, and it is easy to see that the derivative ${U'_1(h_1)\partial h_1/\partial C_p}$ does not depend on the aggregate rate. In our example, changing the aggregate request rate for the second content provider with $U_2(h_2) = h_2$ results in different partition sizes. As the popularity distribution for contents from the first content provider becomes more skewed, \ie\ as $z_1$ increases, the set of popular files decreases in size. Consequently, the dedicated partition size for content provider one decreases as $z_1$ increases. Increasing $\alpha_1$ changes the notion of fairness between the two content providers in favor of the second content provider, and the size of the partition allocated to the first content provider and its hit rate decreases as $\alpha_1$ increases.

Figure~\ref{fig:bwlama_hitrates_shared} repeats the same experiment for the case when some common content is served by both content providers. The cache is partitioned into three slices in this case, one of them storing common content. Very similar behavior as in Figure~\ref{fig:bwlama_hitrates} is observed here.

To understand the fairness notion of the $\alpha$-fair utility functions, we next use the same utility function for both of the content providers, and vary the value of $\alpha$ to see how the hit rates and partition sizes change. Figure~\ref{fig:ubeta} shows the effect of $\alpha$ on hit rates and partition sizes for the case when content providers serve distinct files. As $\alpha$ increases, partition sizes change so that hit rates become closer to each other. This is expected since the $\alpha$-fair utility function realizes the max-min notion of fairness as $\alpha\rightarrow\infty$.

\begin{figure}[t]
\centering
\begin{subfigure}[b]{\linewidth}
  	\centering\includegraphics[scale=0.3]{legend.png}
 \end{subfigure}
 \begin{subfigure}[b]{0.5\linewidth}
  	\centering\includegraphics[scale=0.3]{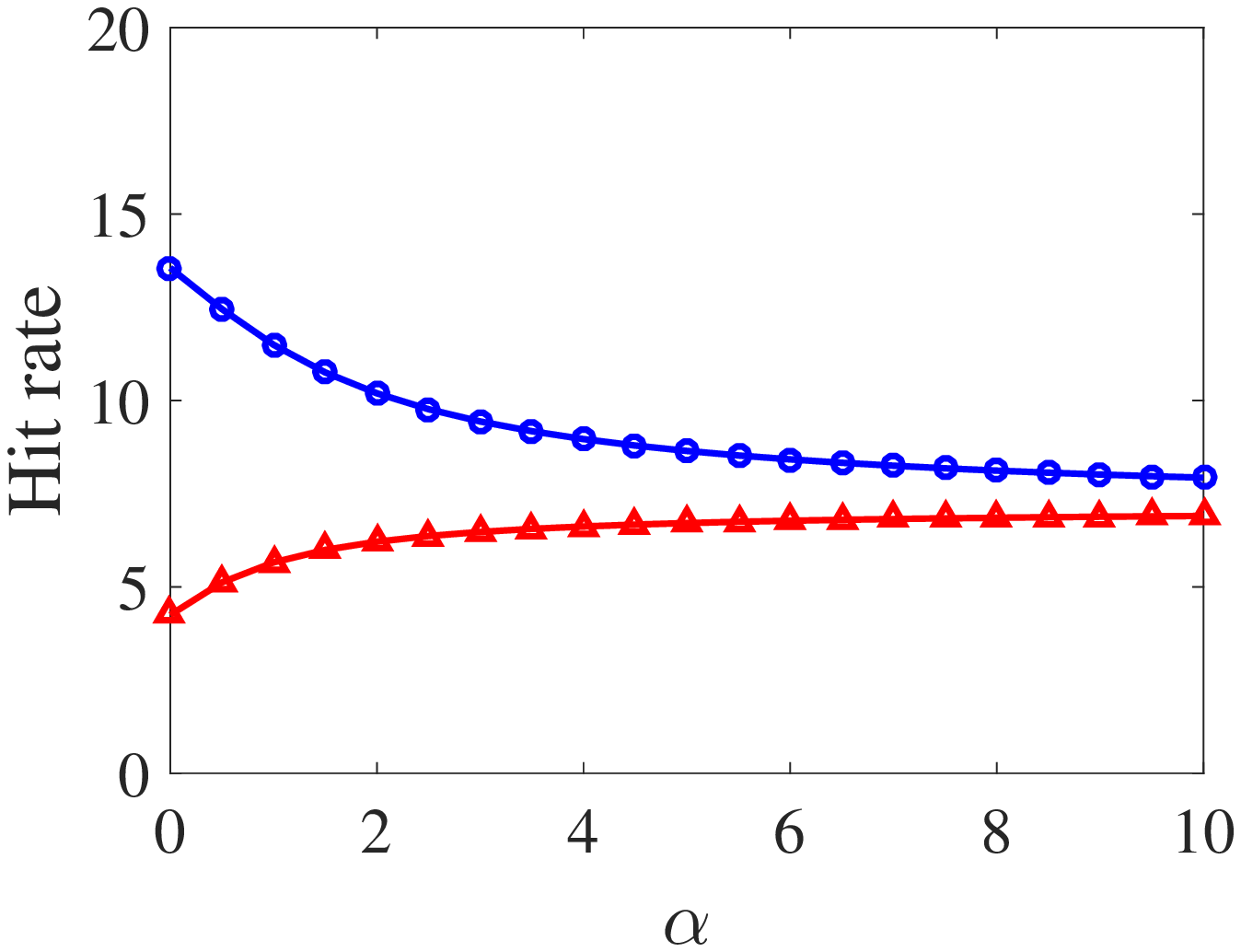}
 \end{subfigure}%
 \begin{subfigure}[b]{0.5\linewidth}
  	\centering\includegraphics[scale=0.3]{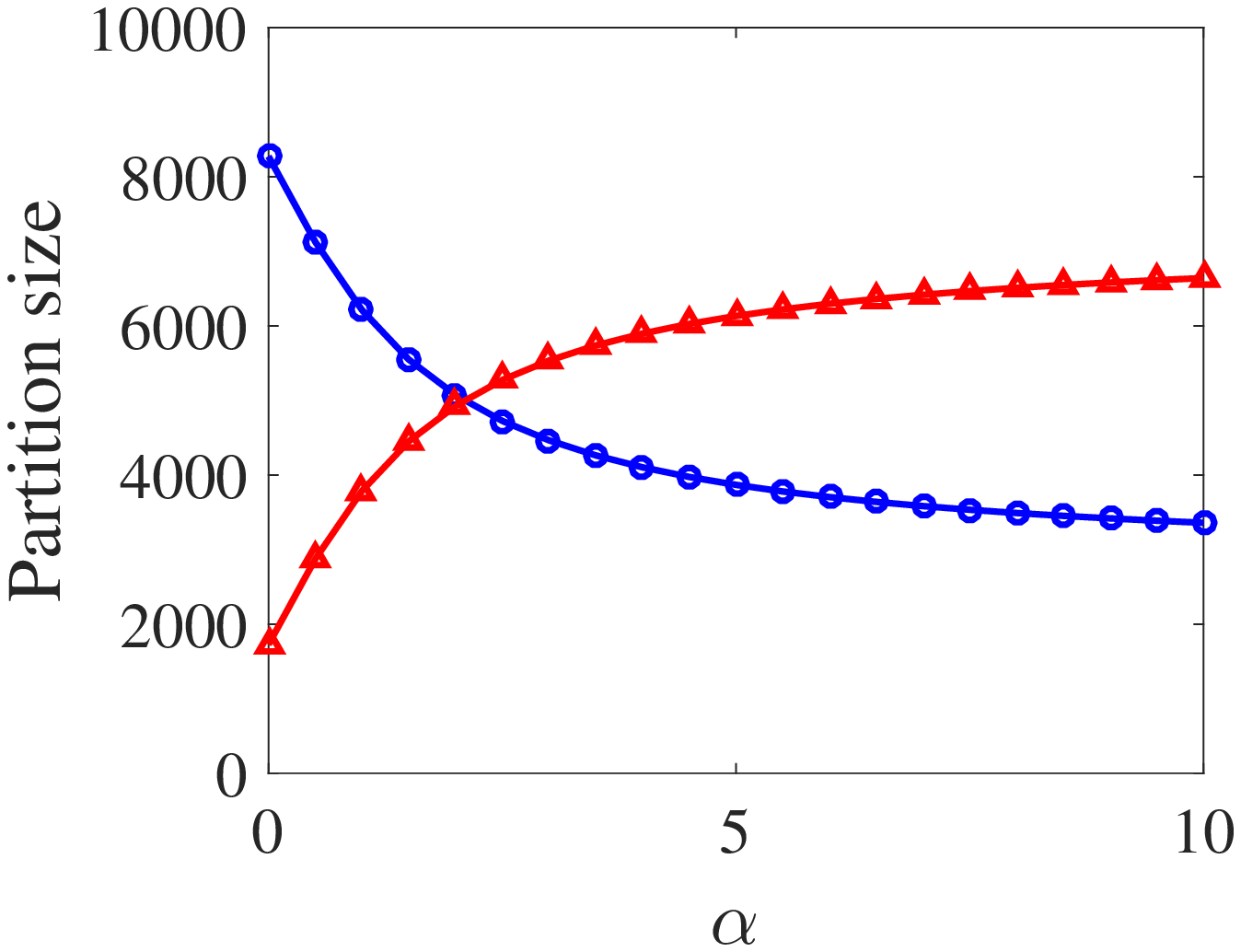}
 \end{subfigure}%
 \caption{$\alpha$-fair resource allocation for content providers serving distinct content. $U_k(h_k) = h_k^{1-\alpha}/(1-\alpha)$.}
    \centering\label{fig:ubeta}
\end{figure}

Figure~\ref{fig:ubeta_shared} shows the changes in resource allocation based on the $\alpha$-fair notion of fairness when common content is served by the content providers.

\begin{figure}[]
\centering
\begin{subfigure}[b]{\linewidth}
  	\centering\includegraphics[scale=0.3]{legend_shared.png}
 \end{subfigure}
 \begin{subfigure}[b]{0.5\linewidth}
  	\centering\includegraphics[scale=0.3]{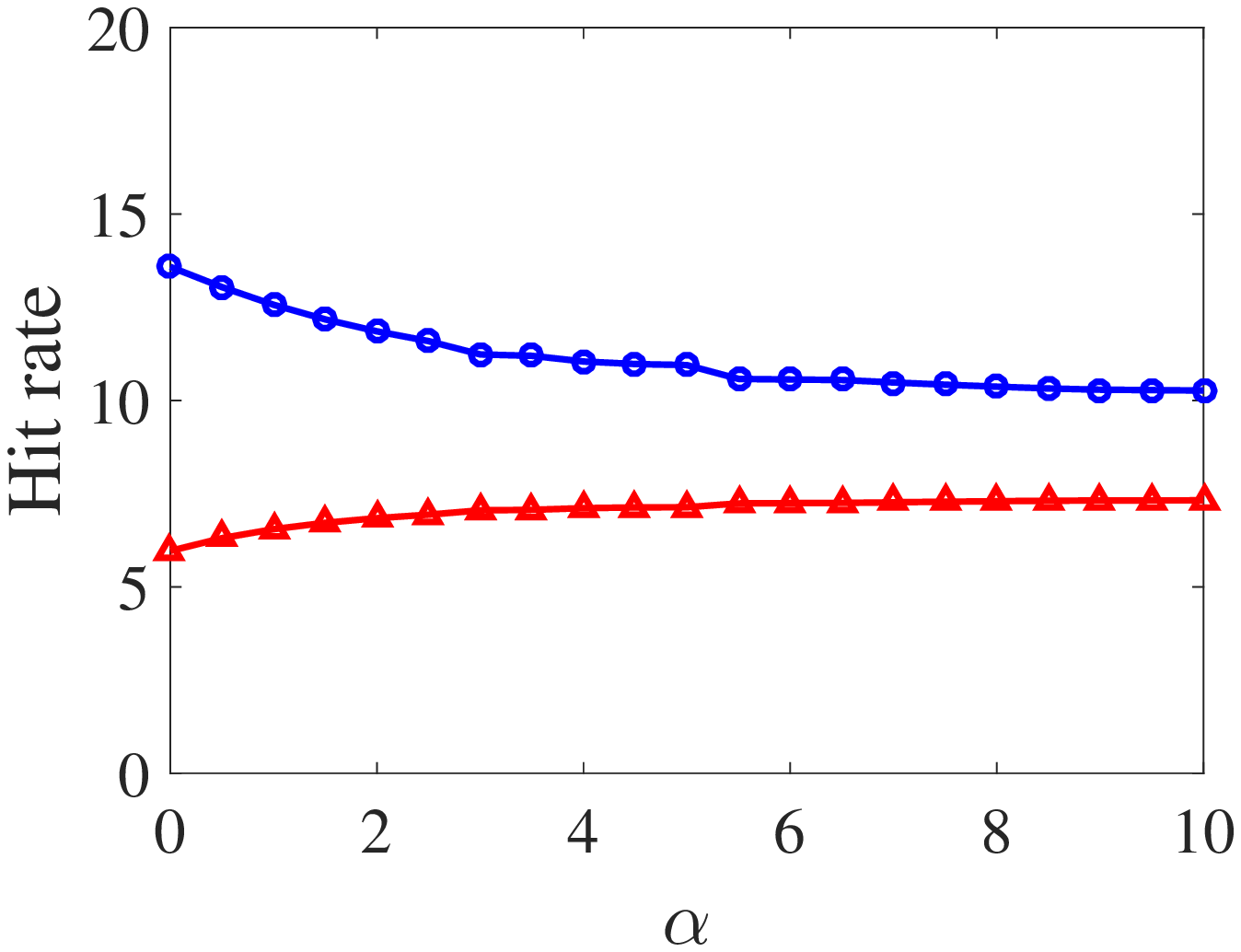}
 \end{subfigure}%
 \begin{subfigure}[b]{0.5\linewidth}
  	\centering\includegraphics[scale=0.3]{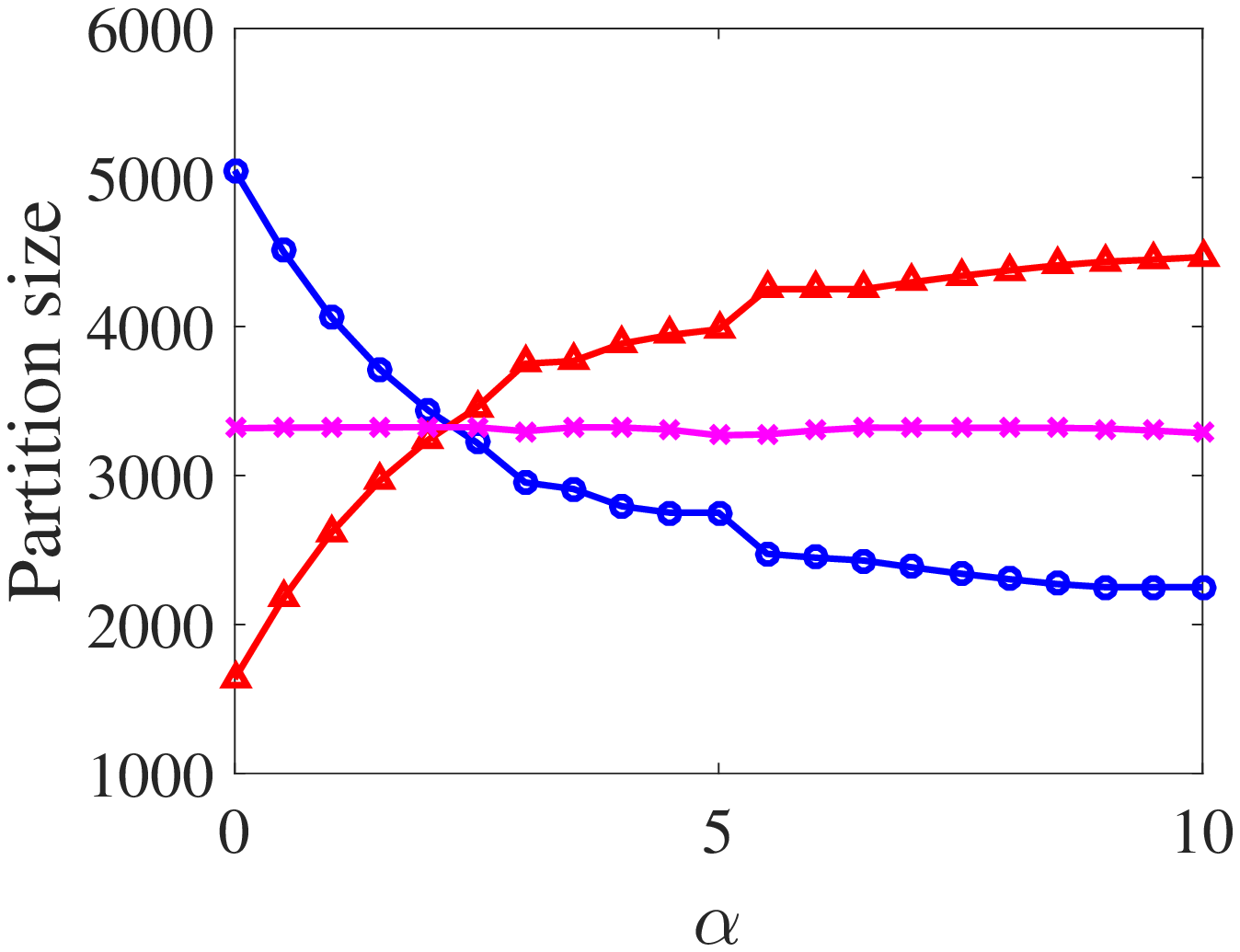}
 \end{subfigure}%
 \caption{$\alpha$-fair resource allocation when some content is served by both content providers. $U_k(h_k) = h_k^{1-\alpha}/(1-\alpha)$.}
    \centering\label{fig:ubeta_shared}
\end{figure}

\subsection{Online Algorithms}
Here, we evaluate the  online algorithms presented in Section~\ref{sec:online} through numerical simulations. Requests are generated according to the parameters presented in the beginning of the section, and the service provider adjusts partition sizes based on the number of hits between iterations. The service provider is assumed to know the utility functions of the content providers. The utility function of the first content provider is fixed to be $U_1(h_1) = \log{h_1}$. We consider three utility functions for the second content provider, namely $U_2(h_2) = h_2$, $U_2(h_2) = \log{h_2}$ and $U_2(h_2) = -1/h_2$.

\begin{figure*}[t]
\centering
 \begin{subfigure}[b]{0.30\linewidth}
  	\centering\includegraphics[scale=0.33]{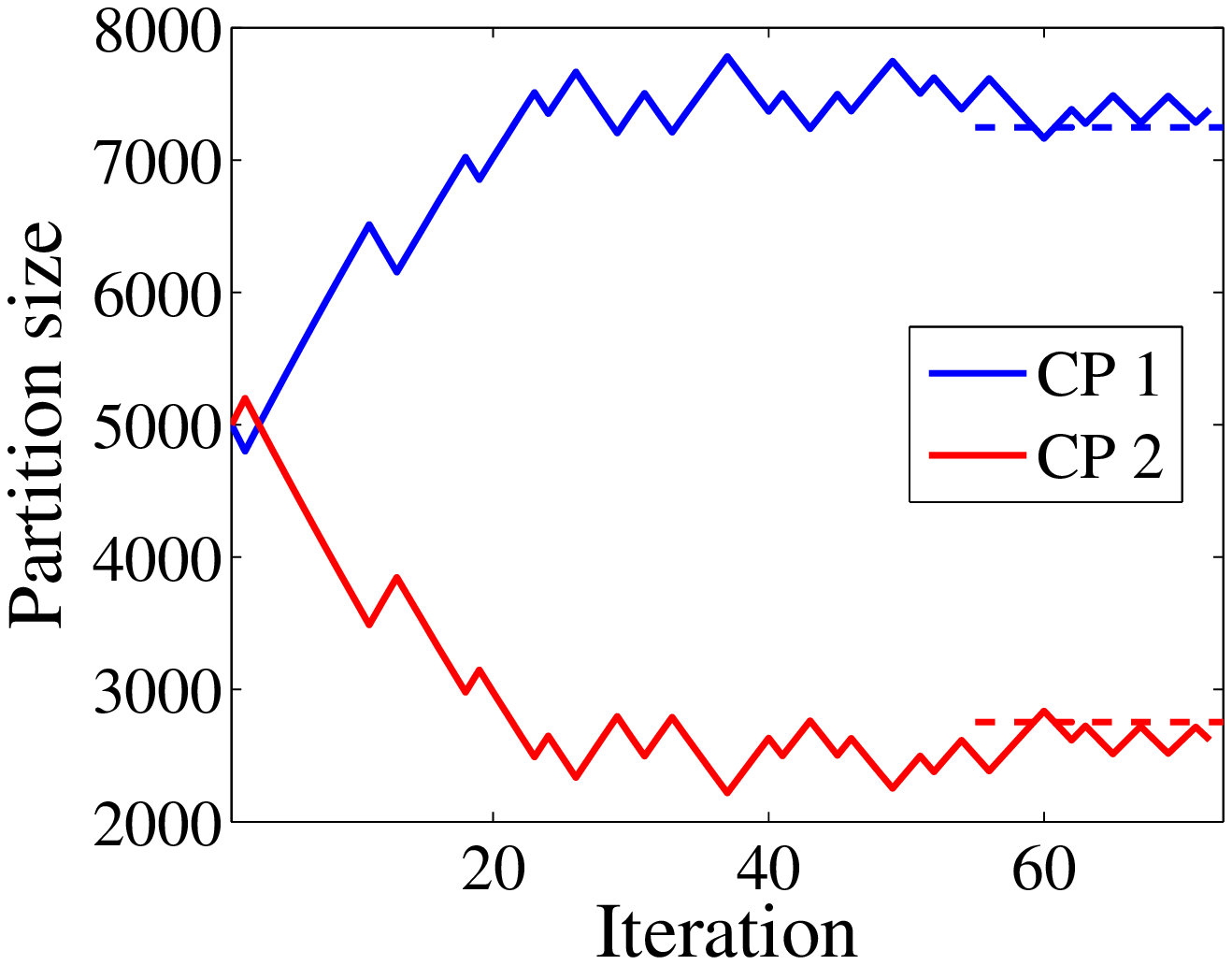}
  	\caption{$U_2(h_2) = h_2$.}
 \end{subfigure}%
 \begin{subfigure}[b]{0.30\linewidth}
  	\centering\includegraphics[scale=0.33]{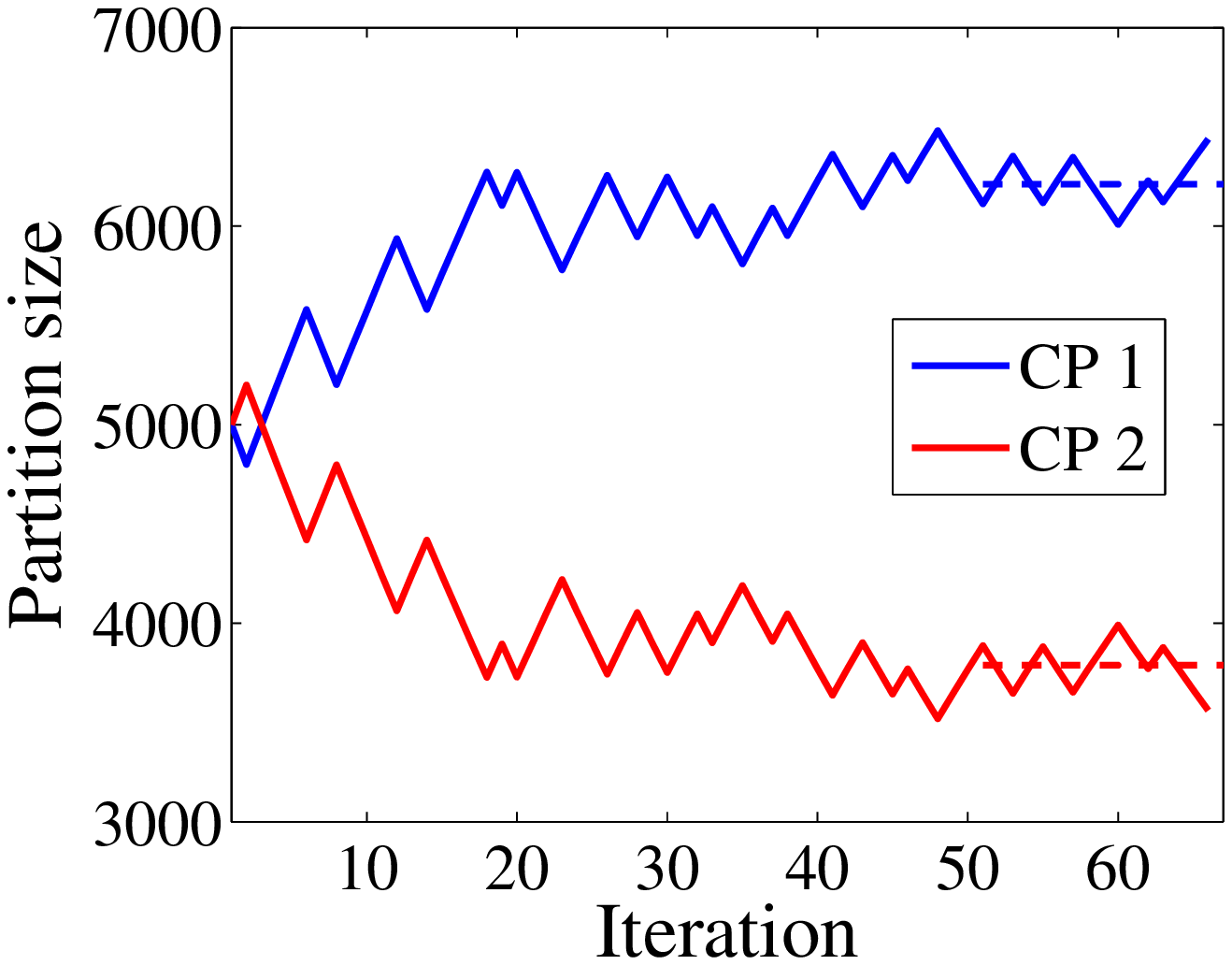}
  	\caption{$U_2(h_2) = \log{h_2}$.}
 \end{subfigure}%
 \begin{subfigure}[b]{0.30\linewidth}
  	\centering\includegraphics[scale=0.33]{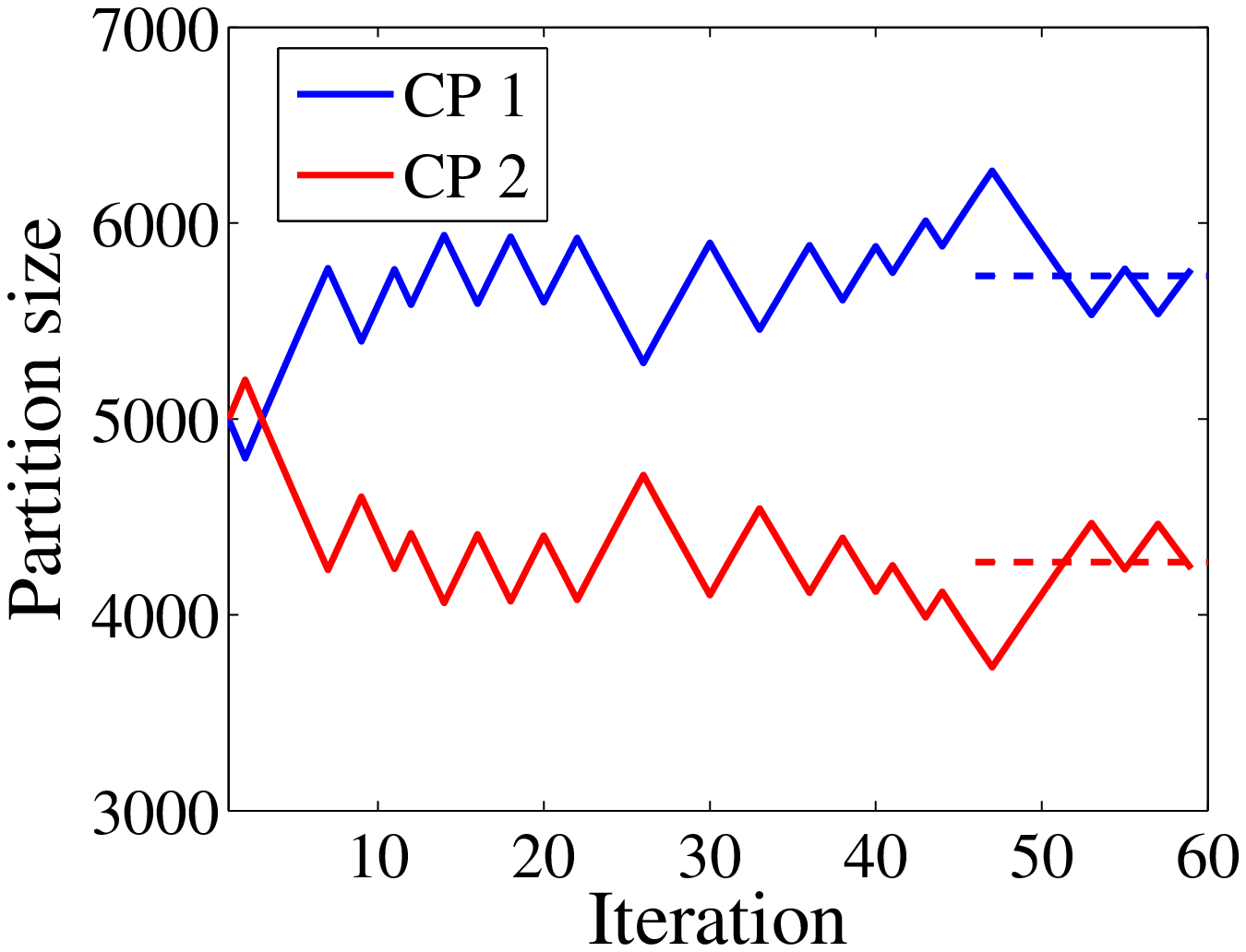}
  	\caption{$U_2(h_2) = -1/h_2$.}
 \end{subfigure}%
 \caption{Convergence of the online algorithm when content providers serve distinct files. $U_1(h_1) = \log{h_1}$.}
    \centering\label{fig:online}
\end{figure*}

We first consider the case where content providers serve distinct files. We initially partition the cache into two equal size slices $C_1 = C_2 = 5000$ and use Algorithm~\ref{alg:online} to obtain the optimal partition sizes. Figure~\ref{fig:online} shows how the partition sizes for the two content providers change at each iteration of the algorithm and that they converge to the optimal values computed from~\eqref{eq:opt}, marked with dashed lines.

Next, we consider the case where some content is served by both content providers. We first partition the cache into three slices of sizes ${C_1=C_2 = 4000}$ and $C_3=2000$, where slice 3 serves the common content, and use Algorithm~\ref{alg:online_shared} to obtain the optimal partitioning. Figure~\ref{fig:online_shared} shows the changes in the three partitions as the algorithm converges to a stable point. For each partition the optimal size computed by~\eqref{eq:opt3} is shown by dashed lines.

\begin{figure*}[t]
\centering
 \begin{subfigure}[b]{0.30\linewidth}
  	\centering\includegraphics[scale=0.33]{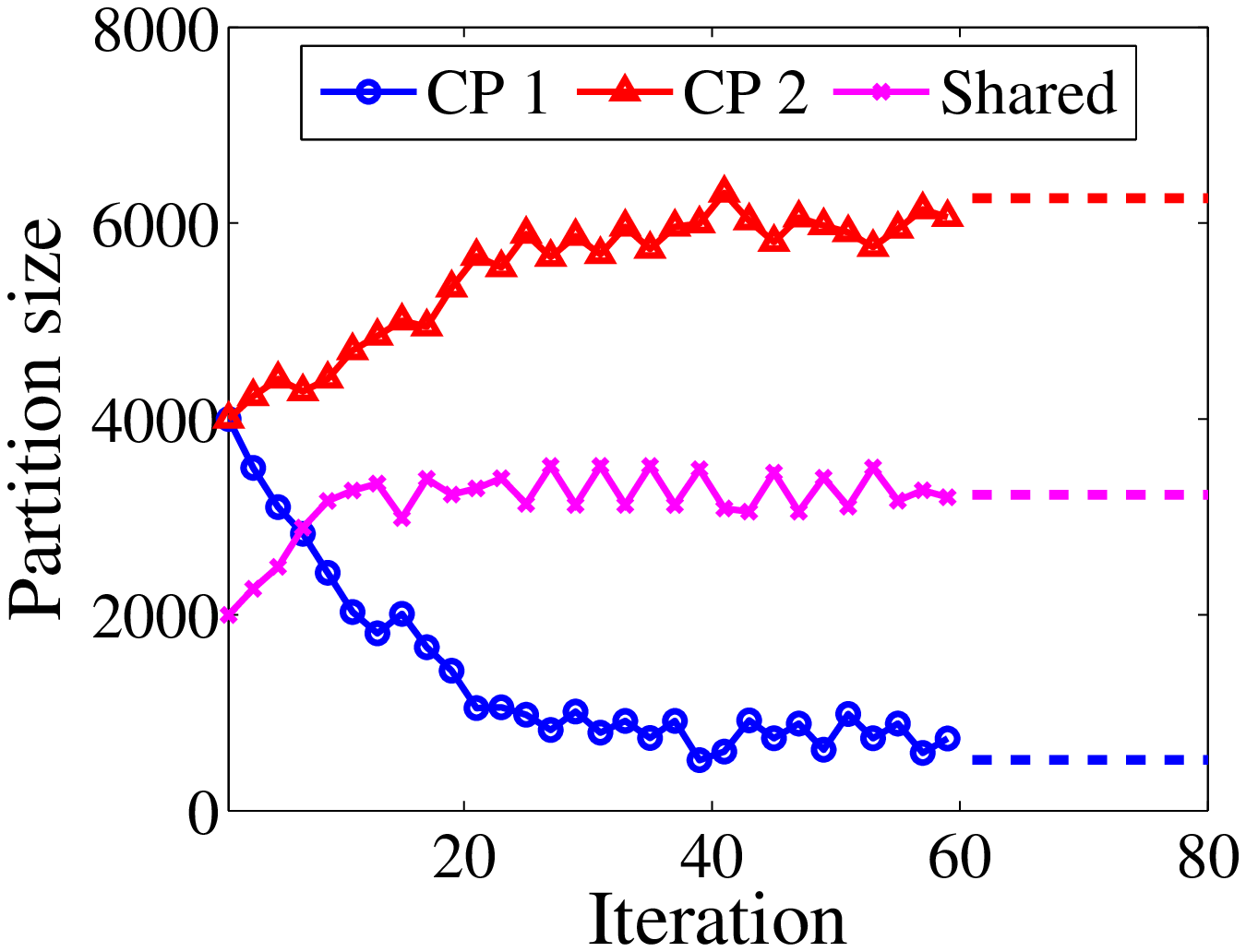}
  	\caption{$U_2(h_2) = h_2$.}
 \end{subfigure}%
 \begin{subfigure}[b]{0.30\linewidth}
  	\centering\includegraphics[scale=0.33]{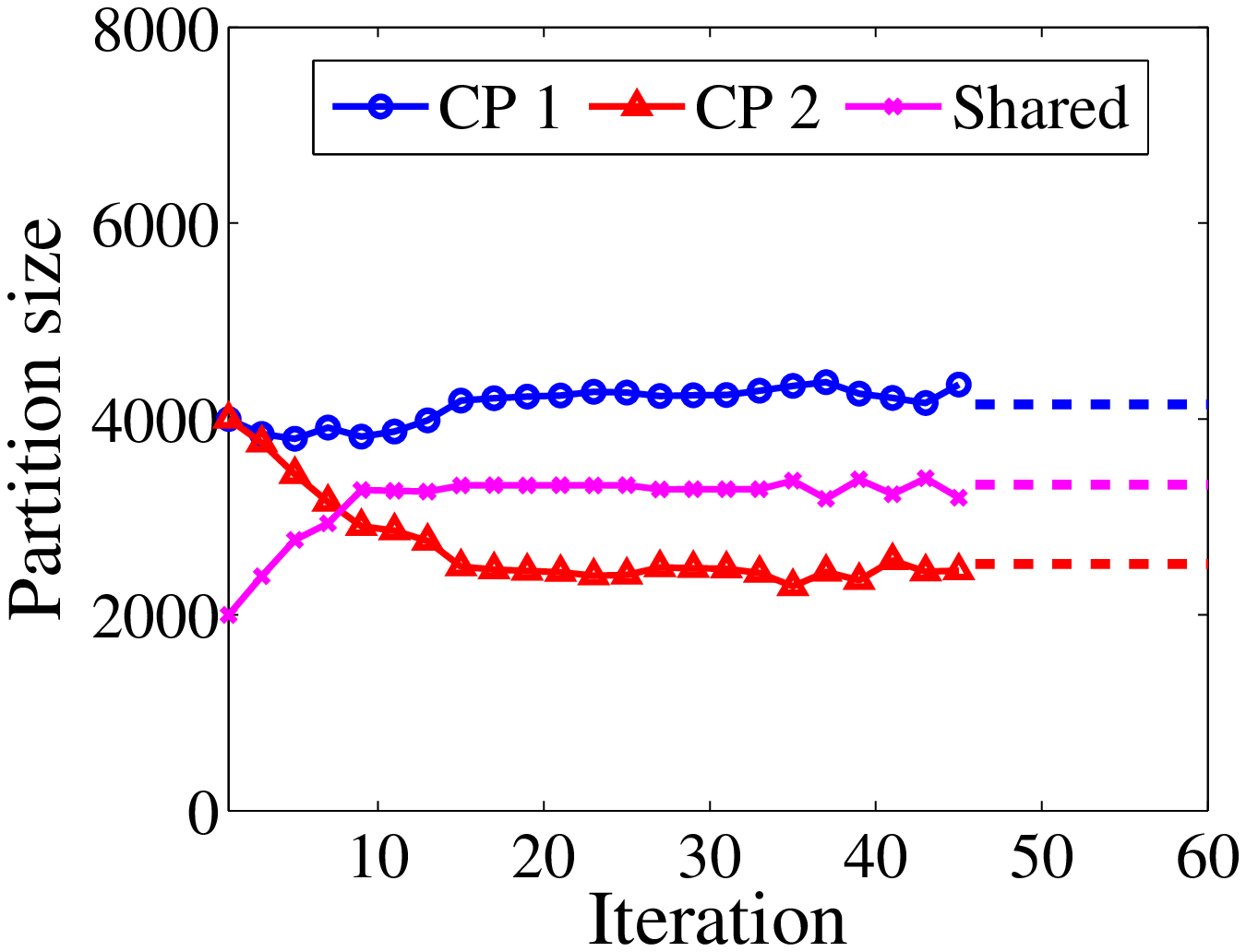}
  	\caption{$U_2(h_2) = \log{h_2}$.}
 \end{subfigure}%
 \begin{subfigure}[b]{0.30\linewidth}
  	\centering\includegraphics[scale=0.33]{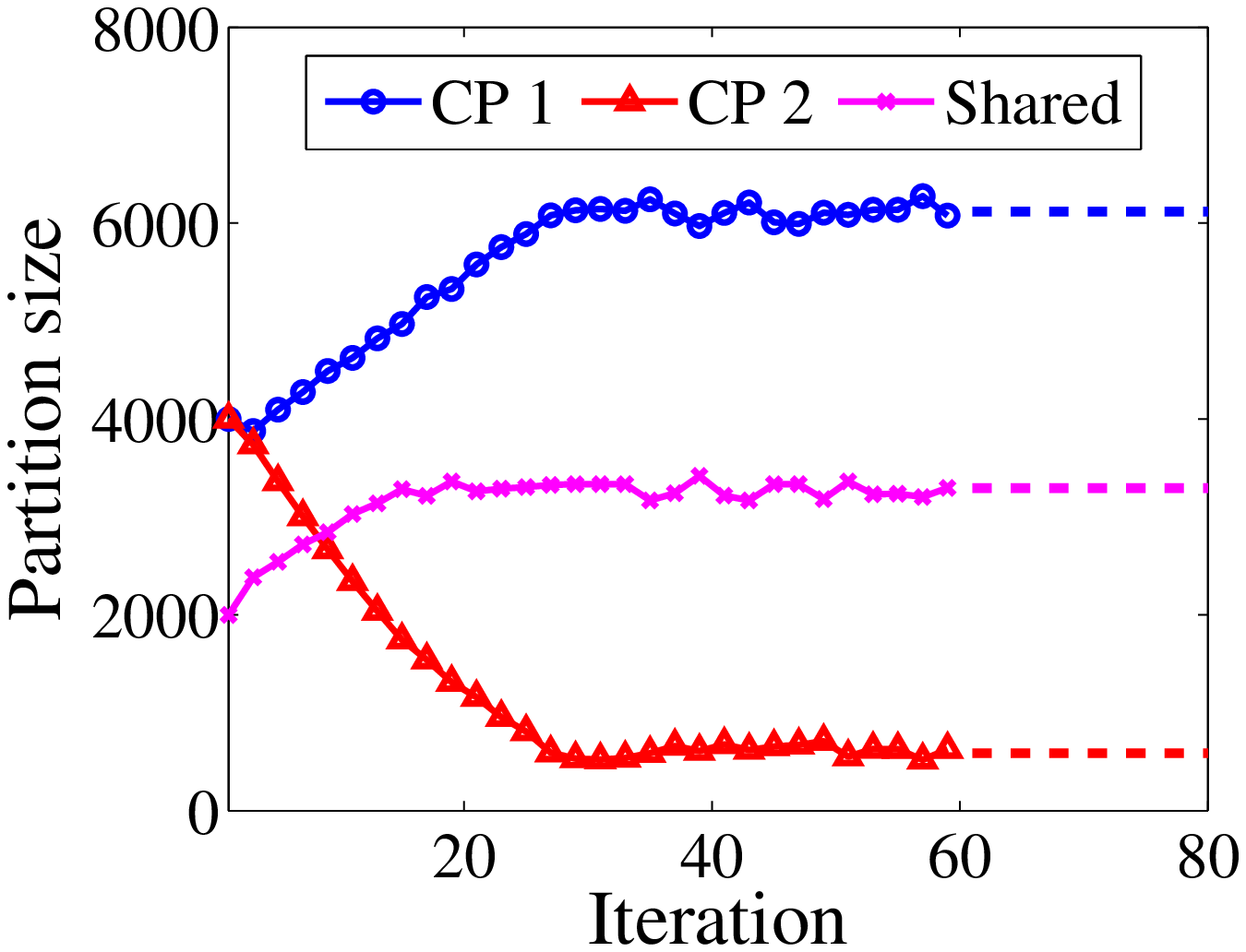}
  	\caption{$U_2(h_2) = -1/h_2$.}
 \end{subfigure}%
 \caption{Convergence of the online algorithm when some content is served by both content providers. $U_1(h_1) = \log{h_1}$.}
    \centering\label{fig:online_shared}
\end{figure*}

%% file: discussion.tex
\section{Discussions and Related Work}
\label{sec:discussion}
In this section, we explore the implications of utility-driven cache partitioning on monetizing caching service and present some future research directions. We end with a brief discussion of the related work.

\vspace*{2pt}
\noindent
{\bf Decomposition.}
The formulation of the problem in Section~\ref{sec:problem} assumes that the utility functions $U_k(\cdot)$ are known to the system. In reality the content providers may not want to reveal their utility functions to the service provider. To handle this case, we decompose optimization problem~\eqref{eq:opt} into two simpler problems.

Suppose that cache storage is offered as a service and the service provider charges content providers at a constant rate $r$ for storage space. Hence, a content provider needs to pay an amount of $w_k = rh_k$ to obtain hit rate $h_k$. The utility maximization problem for content provider $k$ can be written as
\begin{align}
\label{eq:opt_user}
\text{maximize} \quad &U_k(\frac{w_k}{r}) - w_k \\
\text{such that} \quad &w_k \ge 0 \notag
\end{align}

Now, assuming that the service provider knows the vector $\mathbf{w}$, for a proportionally fair resource allocation, the hit rates should be set according to
\begin{align}
\label{eq:opt_network}
\text{maximize} \quad &\sum_{k=1}^{K}{w_k\log{(h_k)}} \\
\text{such that} \quad &\sum_{p}{C_p} = C. \notag
\end{align}

It was shown in~\cite{kelly97} that there always exist vectors $\mathbf{w}$ and $\mathbf{h}$,
such that $\mathbf{w}$ solves~\eqref{eq:opt_user} and $\mathbf{h}$ solves~\eqref{eq:opt_network}; furthermore, the vector $\mathbf{h}$ is the unique optimal solution.

\vspace*{2pt}
\noindent
{\bf Cost and Utility Functions.}
In Section~\ref{sec:online}, we defined a penalty function denoting the cost of using additional storage space. One might also define cost functions based on the consumed network bandwidth.
This is especially interesting in modeling in-network caches with network links that are likely to be congested.

Optimization problems~\eqref{eq:opt} and~\eqref{eq:opt3} use utility functions defined as functions of the hit rate. It is reasonable to define utility as a function of the hit \emph{probability}. Whether this significantly changes the problem, \eg\ in the notion of fairness, is a question that requires further investigation. One argument in support of utilities as functions of hit rates is that a service provider might prefer pricing based on request rate rather than the cache occupancy. Moreover, in designing hierarchical caches a service provider's objective could be to minimize the internal bandwidth cost. This can be achieved by defining the utility functions as $U_k = -P_k(m_k)$ where $P_k(m_k)$ denotes the cost associated with miss rate $m_k$ for content provider $k$.


%% file: related_work.tex

\vspace*{2pt}
\noindent
{\bf Related Work.}
Internet cache management issues have been extensively studied in the context of web caching (e.g., see~\cite{Che01,feldman02} and references therein). In this context, biased replacement policies for different kinds of  content classes~\cite{kelly99} and differentiated caching services via cache partitioning~\cite{ko03,lu04} haven been proposed and studied. None of these studies explicitly deal with the cache allocation problem among multiple content providers. The emergence of content-oriented networking has renewed research interests in cache management issues for content delivery, especially in the design of cache replacement policies for  the content-oriented architecture~\cite{borst10,Carofiglio11,Zhang13,Michelle14}. The cache allocation problem among content providers has attracted relatively little attention. Perhaps most closely related to our work is the study in~\cite{Hoteit15} where a game-theoretic cache allocation approach is developed. This approach requires the content providers to report the true demands from their content access. 
In contrast, we develop a general utility maximization framework for studying the cache allocation problem. Since its first proposal by Kelly~\etal~\cite{kelly97}, the network utility maximization framework has been applied to a variety of networking problems  from stability analysis of queues~\cite{eryilmaz07} to the study of fairness in network resource allocation~\cite{neely10}. A utility maximization framework for caching policies was developed in~\cite{dehghan16} to provide differentiated services to content. This framework was adopted by~\cite{ICN16} to study the cache resource allocation problem in an informal and heuristic manner. We make precise statements to support the observations in~\cite{ICN16}. In this respect, our contribution lies in establishing the key properties of CP utilities as a function of cache sizes and in postulating cache partitioning as a basic principle for cache sharing among content providers. Furthermore, we develop decentralized algorithms to implement utility-driven cache partitioning, and prove that they converge to the optimal solution.

%% file: conclusion.tex
\section{Conclusion}
\label{sec:conclusion}
We proposed utility-based partitioning of a cache among content providers, and formulated it as an optimization problem with constraints on the service providers cache storage size. Utility-driven cache partitioning provides a general framework for managing a cache with considerations of fairness among different content providers, and has implications on market economy for service providers and content distributors. We considered two scenarios where 1) content providers served disjoint sets of files, or 2) some content was served by multiple content providers. We developed decentralized algorithms for each scenario to implement utility-driven cache partitioning in an online fashion. These algorithms adapt to changes in request rates of content providers by dynamically adjusting the partition sizes. We theoretically proved that these algorithms are globally stable and converge to the optimal solution, and through numerical evaluations illustrated their efficiency.

%% file: appendices.tex

\appendix

\section{Hit rate is a concave and increasing function of cache size.}
\label{sec:hitrate_concavity}
\input{proof_h_concave_increasing}

\section{Proof of Theorem~\ref{THM:SHARED_FAGIN}}
\label{sec:shared_Fagin}
\input{proof_Fagin}



\section{Proof of Theorem~\ref{THM:3OVER2}}
\label{sec:3over2}
\input{proof_3over2}

\section{Proof of Theorem~\ref{THM:UNIQUE}}
\label{sec:proof_obj_shared}
\input{proof_obj_shared.tex}

\section{Partitioning and Probabilistic Routing is Sub-optimal.}
\label{sec:proof_partition}
\input{proof_partition.tex}

\section{Proof of Theorem~\ref{THM:GA}}
\label{sec:lyapunov}
\input{proof_lyapunov.tex}


%% file: proof_h_concave_increasing.tex
\begin{lemma}
The hit rate $h_k$ is a concave and strictly increasing function of $C_k$.
\end{lemma}

\begin{proof}
Differentiating (\ref{hkCk}) and (\ref{Ck}) w.r.t. $C_k$ gives
\begin{eqnarray*}
\frac{dh_k(C_k)}{dC_k}&=& \lambda_k^2\sum_{i=1}^{n_k} p_{k,i}^2 e^{-\lambda_kp _{k,i} T_k(C_k)}\frac{dT_k(C_k)}{dC_k} \\
1&=&\lambda_k \sum_{i=1}^{n_k} p_{k,i}  e^{-\lambda_k p_{k,i} T_k(C_k)}  \frac{dT_k(C_k)}{dC_k}.
\end{eqnarray*}
The latter equation implies that $dT_k(C_k)/dC_k>0$, which in turn implies from the former that $dh_k(C_k)/dC_k> 0$. This proves that $h_k(C_k)$ is
strictly increasing in $C_k$. Differentiating now the above equations w.r.t. $C_k$ yields
\begin{eqnarray}
\frac{1}{\lambda_k^2} \frac{d^2h_k(C_k)}{dC^2_k}&=&\sum_{i=1}^{n_k} p_{k,i}^2  e^{-\lambda_k p_{k,i} T_k(C_k)}g_{k,i} \nonumber\\
0&=&  \sum_{i=1}^{n_k} p_{k,i}  e^{-\lambda_k p_{k,i} T_k(C_k)} g_{k,i}\label{deriv2} 
\end{eqnarray}
with $g_{k,i}:= d^2T_k(C_k)/dC^2_k - \lambda_k p_{k,i} (dT_k(C_k)/dC_k)^2$.
Assume without loss of generality that $0\leq p_{k,1}\leq \cdots\leq p_{k,n_k}\leq 1$.
(\ref{deriv2}) implies that there exists $1\leq l \leq n_{k}$ such that $ g_{k,i}\geq 0$ for $i=1,\ldots,l$ and  $ g_{k,i}\leq 0$ for $i=l+1,\ldots,n_k$.
Hence,
\begin{eqnarray*}
\frac{1}{\lambda_k^2}\frac{d^2h_k(C_k)}{dC_k^2}&\leq &\sum_{i=1}^{l}  p_{k,i}  e^{-\lambda_k p_{k,i} T_k(C_k)}g_{k,i} \\
&&\quad + \sum_{i=l+1}^{n_k} p_{k,i}^2 e^{-\lambda_k p_{k,i} T_k(C_k)}g_{k,i}\\
&=& \sum_{i=l}^{n_k}  p_{k,i}  e^{-\lambda_k p_{k,i} T_k(C_k) }    g_{k,i}(1-p_{k,i})\leq 0.
\end{eqnarray*}
This proves that $h_k(C_k)$ is concave in $C_k$.
\end{proof}

%% file: proof_Fagin.tex

\begin{proof}
We first construct a CDF $F$ from the CP specific CDFs, $\{F_k\}$.  When the providers share a single cache, documents are labelled $1, \ldots,B n$, so that documents $B_{k-1} + 1, \ldots, B_k$ are the $b_k$ documents with service provider $k$. Denote $A_k := \sum_{j=1}^{k} a_j$ with $A_0 = 0$ in what follows.
Define
\begin{multline}
F(x) = \sum_{j=1}^K   \Bigl(A_{k-1} + a_k F_k\Bigl(\frac{B}{b_k}x - \frac{B_{k-1}}{b_k}\Bigr)\Bigr)  \\
   \times {\bf 1}\left\{ \frac{B_{k-1} }{B} \leq x \leq  \frac{B_k} {B}\right\}.
\end{multline}
Let
\[
p^{(n )}_i:=F\left(\frac{i}{Bn}\right) - F\left(\frac{i-1}{Bn}\right), \quad i=1,\ldots,nB_K
\]
It is easy to see that 
\[
p^{(n)}_{B_{k-1} + i} = a_k \,p^{(n)}_{k,i}, \quad i = 1,\dotsc , b_k;\,k=1,\dotsc ,K.
\]

Note that $F$ may not be differentiable at $x\in\{B_1/B, B_2/B, \dotsc B_{K-1}/B\}$ and, hence, we cannot apply the result of \cite[Theorem 1]{Fagin1977} directly to our problem.

Let 
\[
\beta^{(s)}(n,\tau_0)= 1-\frac{1}{B n} \sum_{i=1}^{B n} \left( 1-p_i^{(n)}\right)^{n\tau_0}
\]
be the fraction of  documents  in the cache. Here, $n\tau_0$ corresponds to the window size in \cite{Fagin1977}.

We have 
\begin{align}
 \beta^{(s)}(n,\tau_0) & = 1 - \frac{1}{B n} \sum_{k=1}^K \sum_{i=1}^{b_k n}  \bigl(1-p^{(n)}_{B_{k-1}+i}\bigr)^{n\tau_0 } \nonumber \\
 & = 1-\frac{1}{Bn} \sum_{k=1}^K \sum_{i=1}^{b_k n}  \bigl(1-a_k p^{(n)}_{k,i}\bigr)^{n\tau_0 } 
   \label{sum}.
\end{align}
 We are interested in $\beta^{(s)}=\lim_{n\to \infty}\beta^{(s)}(n,\tau_0)$.
 
 By applying the result in \cite[Theorem 1]{Fagin1977}, we find that 
\begin{multline}
\lim_{n\to \infty}  \frac{1}{Bn} \sum_{i=1}^{b_k n} \bigl(1-a_k p^{(n)}_{k,i}\bigr)^{n\tau_0 }  = \\
\frac{b_k}{B}\int_0^1 e^{-\tau_0 a_k BF_k^\prime(x)/b_k}dx
\end{multline}
for $k=1,\dotsc ,K$, so that
\[
  \beta^{(s)} = 1-\sum_{k=1}^K \frac{b_k}{B} \int_0^1 e^{- a_k F_k^\prime (x)\tau_0 B/b_k} dx .
 \]
Equation \eqref{eq:mu-shared} is derived in the same way.

Last, it follows from Theorems 2 and 4 in \cite{Fagin1977} that $\mu^{(s)}$ is the limiting aggregate miss probability under LRU as $n\rightarrow \infty$.
\end{proof}

%% file: proof_3over2.tex

\begin{proof}
The optimization problems under strategies 2 and 3 are 
\begin{align*}
\min_{\tau_1,\tau_2} \quad &\sum_{k=1}^2 (a_{0,k}\mu^{(s2)}_{0,k}(\tau_k) + a_k\mu^{(s2)}_k(\tau_k)) \\
\text{s.t.} \quad &\sum_{k=1}^2 (\beta^{(s2)}_{0,k}(\tau_k) + \beta^{(s2)}_k(\tau_k)) \le \beta, \\
& \beta^{(s2)}_{0.k},\,\beta^{(s2)}_k \ge 0, \quad k=1,2. 
\end{align*}
and
\begin{align}
\min_{\tau_k} \quad &\sum_{k=1}^3 a_k \mu^{(s3)}_k(\tau_k) \notag\\
\text{s.t.} \quad &\sum_{k=1}^3\beta^{(s3)}_k(\tau_k)  \beta, \label{eq:s3constraint}\\
& \beta^{(s3)}_k \ge 0, \quad k=1,2,3. \notag
\end{align}
where $\lambda_0 = \lambda_{0,1}+\lambda_{0.2}$, $a_{0,k} = \lambda_{0,k}/\sum_{k=1}^3\lambda_k$, and $a_k = \lambda_k /\sum_{k=1}^3\lambda_k$,
\begin{align*}
\mu^{(s2)}_{0,k}(\tau) &=  \int_0^1 F'_0(x) e^{-\frac{a_{0,k}(b_0+b_k)}{(a_{0,k}+a_k)b_0}F'_0(x)\tau}dx \\ 
\mu^{(s2)}_k(\tau) &=  \int_0^1 F'_k(x) e^{-\frac{a_{0,k}(b_0+b_k)}{(a_{0,k}+a_k)b_k}F'_k(x)\tau}dx \\
\beta^{(s2)}_{0,k}(\tau) &=   \frac{b_0}{2b_0+b_1+b_2}\bigl(1 -  \int_0^1  e^{-\frac{a_{0,k}(b_0+b_k)}{(a_{0,k}+a_k)b_0}F'_0(x)\tau}dx\bigr) \\
\beta^{(s2)}_k(\tau) &=   \frac{b_k}{2b_0+b_1+b_2}\bigl(1 -  \int_0^1  e^{-\frac{a_k(b_0+b_k)}{(a_{0,k}+a_k)b_k}F'_0(x)\tau}dx\bigr)
\end{align*}
and
\begin{align*}
\mu^{(s3)}_k(\tau) &=  \int_0^1 F'_k(x) e^{-F'_k(x)\tau}dx \\ 
\beta^{(s3)}_k(\tau) &=   \frac{b_k}{b_0+b_1+b_2}\bigl(1 -  \int_0^1  e^{-F'_k(x)\tau}dx\bigr)
\end{align*}
We make two observations:
\begin{itemize}
\item $\mu^{(s2)}_{0,k}(\tau ) = \mu^{(s3)}_0(\frac{a_{0,k}(2b_0+b_1+b_2)}{(a_{0,k}+a_k)b_0})\tau$, $k=1,2$,
\item $\mu^{(s3)}_0 (\tau )$ is a decreasing function of $\tau$.
\end{itemize}
Let $\mu^{(s2)*}$ denote the minimum miss probability under strategy 2, which is achieved with $\tau^*_1$ and $\tau^*_2$.  Let $\mu^{(s3)}(\tau_1,\tau_2,\tau_3)$ denote the miss probability under strategy 3 where $\tau_k$, $k=1,2,3$ satisfy \eqref{eq:s3constraint}.   
Set $\tau_k=\frac{a_{0,k}(2b_0+b_1+b_2)}{(a_{0,k}+a_k)b_0}\tau^*_k$, $k=1,2$ for strategy 3 and allocate provider $k$ a cache of size $\beta^{(s2)}_k$ under strategy 3 for its non-shared content.  The aggregate miss probability for non-shared content is then the same under the two strategies and given by $\mu^{(s2)}_1(\tau^*_1) + \mu^{(s2)}_2(\tau^*_2)$.

Under strategy 2, the amount of shared content stored in the cache is $\beta_- = \beta^{(s2)*}_{0,1} + \beta^{(s2)*}_{0,2}$.  We allocate a cache of that size to the shared content under strategy 3 and wlog assume that $\tau^*_1 \ge \tau^*_2$ Themiss probability over all shared content under strategy 2 is
\[
\begin{split}
\sum_{k=1}^2 \frac{a_{0,k}}{a_{0,1}+a_{0.2}} \mu^{(s2)}_{0,k}(\tau^*_k) & \ge \mu^{(s2)}_{0,k}(\tau^*_1) \\
  & = \mu^{(s3)}_0(\tau_1)
\end{split}
\]
 Note that strategy 3 requires only a cache of size $\beta^{(s2)}_{0,k} < \beta_-$ to achieve a smaller miss probability for the shared content than strategy 2 can realize.  Adding additional storage to the shared partition can only decrease the hit probability further, thus proving the theorem.

\end{proof}

%% file: proof_obj_shared.tex
\begin{proof}
Let $P$ and $\mathbf{C} = (C_1, \ldots, C_P)$ denote the number of partitions and the vector of partition sizes, respectively. The hit rate for content provider $k$ in this case can be written as
\[h_k(\mathbf{C}) = \sum_{p=1}^{P}{\sum_{i\in V_p}{\lambda_{ik}(1 - e^{-\lambda_iT_p})}},\]
where $V_p$ denotes the set of files requested from partition $p$, and $\lambda_i = \sum_{k}{\lambda_{ik}}$ denotes the aggregate request rate for file $i$.

We can re-write the expression for $h_k$ as the sum of the hit rates from each partition, since distinct files are requested from different partitions. We have
\[h_k(\mathbf{C}) = \sum_{p=1}^{P}{h_{kp}(C_p)},\]
where $h_{kp}(C_p)$ denotes the hit rate for files requested from partition $p$ from content provider $k$. Since $h_{kp}$ is assumed to be a concave increasing function of $C_p$, $h_k$ is sum of concave functions and hence is also concave.
\end{proof}

%% file: proof_partition.tex
\begin{lemma}
Partitioning a cache, and probabilistically routing content requests to different partitions is sub-optimal.
\end{lemma}

\begin{proof}
Assume we partition the cache into two slices of size $C_1$ and $C_2$, and route requests to partition one with probability $p$, and with probability $1-p$ route to partition two. Let $h_P$ denote the hit rate obtained by partitioning the cache. From concavity of~$h(C)$ we have
\begin{align}
h_P &= ph(C_1) + (1-p)h(C_2) \notag\\
    &\le ph(C) + (1-p)h(C) = h(C).\notag
\end{align}
\end{proof}

%% file: proof_lyapunov.tex
\begin{proof}
We first note that since $W(\mathbf{C})$ is a strictly concave function, it has a unique maximizer $\mathbf{C}^*$.
Moreover $V(\mathbf{C}) = W(\mathbf{C}^*) - W(\mathbf{C})$ is a non-negative function and equals zero only at $\mathbf{C} = \mathbf{C}^*$.
Differentiating $V(\cdot)$ with respect to time we obtain
\begin{align*}
\dot{V}(\mathbf{C}) &= \sum_{k}{\frac{\partial V}{\partial h_k}\dot{h_k}} \\
&= -\sum_{k}{\left( U'_k(h_k) - P'(\sum_{k}{C_k} - C)\frac{\partial C_k}{\partial h_k} \right) \dot{h_k}}.
\end{align*}

For $\dot{h_k}$ we have
\[\dot{h_k} = \frac{\partial h_k}{\partial C_k}\dot{C_k}.\]

From the controller for $C_k$ we have
\[\dot{C_k} = \gamma_k \left( U'_k(h_k)\frac{\partial h_k}{\partial C_k} - P'(\sum_{k}{C_k} - C) \right).\]

Since $\frac{\partial h_k}{\partial C_k} \ge 0$, we get
\begin{align*}
\dot{V}(\mathbf{C}) &= -\sum_{k}{\gamma_k \frac{\partial h_k}{\partial C_k}\left( U'_k(h_k) - P'(\sum_{k}{C_k} - C)\frac{\partial C_k}{\partial h_k} \right)^2} \\
&\le 0.
\end{align*}

Therefore, $V(\cdot)$ is a Lyapunov function, and the system state will converge to $\mathbf{C}^*$ starting from any initial condition. A description of Lyapunov functions and their applications can be found in~\cite{srikant13}.
\end{proof}